\documentclass[journal]{IEEEtran}
\usepackage{amsthm}
\usepackage{amssymb}
\usepackage{graphicx}
\usepackage{amsmath}
\usepackage{newtxmath}
\usepackage{comment}
\usepackage[cal=cm]{mathalpha} 
\usepackage{tablefootnote} 
\usepackage{threeparttable} 
\usepackage{bm}
\newtheorem{theorem}{Theorem}{}
{}
{}
{}
{}
{}

\usepackage{algorithm}
\usepackage{algpseudocode}
\usepackage{graphics}
\usepackage{epsfig}
\usepackage{multirow}
\usepackage{caption}
\usepackage{array}
\usepackage{cite}
\usepackage{stfloats}
\usepackage{booktabs}
\usepackage{color}
\usepackage{tabu} 
\usepackage{stfloats}
\usepackage{mathtools}
\usepackage{subcaption}
\usepackage{stmaryrd} 
\makeatletter
\let\sum\relax
\let\prod\relax
\DeclareSymbolFont{largesymbols}{OMX}{cmex}{m}{n}
\DeclareMathSymbol{\sum}{\mathop}{largesymbols}{"50}
\DeclareMathSymbol{\prod}{\mathop}{largesymbols}{"51}
\usepackage[cal=cm]{mathalpha} 
\usepackage[colorlinks,linkcolor=blue,anchorcolor=blue,citecolor=blue,bookmarks=true]{hyperref}
\usepackage[capitalize]{cleveref} 
\crefname{figure}{Fig.}{Figs.}   
\Crefname{figure}{Fig.}{Figs.}   

\DeclareMathOperator*{\argmax}{argmax}
\DeclareMathOperator*{\argmin}{argmin}

\DeclareMathOperator*{\var}{var}
\DeclareMathOperator*{\cov}{cov}

\DeclareMathOperator*{\diag}{\text{diag}}

\DeclareMathOperator*{\rect}{\text{rect}}

\begin{document}
	\title{Tensor Decomposition-Based Wireless Sensing for MIMO-OFDM ISAC via Flexible Spatial-Temporal-Spectral Optimization}
	
	\author{\IEEEauthorblockN{
			Chengzhi Ye, \IEEEmembership{Student Member,~IEEE}, 
			Ruoyu Zhang, \IEEEmembership{Senior Member,~IEEE}, 
			Lei Yao, Xinrong Guan, Yu Zhang, 
			Wen Wu, 
			\IEEEmembership{Senior Member,~IEEE}, 
			Rui Zhang, \IEEEmembership{Fellow,~IEEE}
		}
		
		\thanks{
			Chengzhi Ye, Ruoyu Zhang, Lei Yao and Wen Wu are with the Key Laboratory of Near-Range RF Sensing ICs and Microsystems (NJUST), Ministry of Education, School of Electronic and Optical Engineering, Nanjing University of Science and Technology, Nanjing 210094, China. Chengzhi Ye and Lei Yao are also with Qian Xuesen College, Nanjing University of Science and Technology, Nanjing 210094, China (e-mail  qxsycz8166@njust.edu.cn; ryzhang19@njust.edu.cn;  yaolei@njust.edu.cn; wuwen@njust.edu.cn).		
			{Xinrong Guan is with the College of Communications
			Engineering, Army Engineering University of PLA, Nanjing 210007, China
			(e-mail: guanxr@aliyun.com).}
			{Yu Zhang  is with The Sixty-Third Research Institute,
			National University of Defense Technology, Nanjing 210007, China. (e-mail: zhyu63@163.com).
			Rui Zhang is with the Department of Electrical and
			Computer Engineering, National University of Singapore, Singapore 117583
			(e-mail: elezhang@nus.edu.sg).}
			\textit{(Corresponding author: Ruoyu Zhang)}
		}
		\vspace{-1.75em}
	}
	
	
	\maketitle
	
	

	%
	\IEEEpeerreviewmaketitle

	\begin{abstract}
		Integrated sensing and communication (ISAC) is regarded as a key enabling technique in future 6th-generation (6G) mobile communication systems. However, existing multi-input multi-output (MIMO) orthogonal frequency division multiplexing (OFDM) ISAC designs generally rely on the fixed-position antennas and fixed allocation of time-frequency resources, thereby limiting the degrees of freedom of wireless sensing along the spatial-temporal-spectral dimensions. 
		In this paper, we propose a novel wireless sensing framework for MIMO-OFDM ISAC systems with flexible spatial-temporal-spectral optimization and propose a tensor decomposition-based approach to estimate target parameters,  including azimuth/elevation angles, ranges, and velocities. 
		Specifically, we first establish a monostatic wireless sensing model for MIMO-OFDM ISAC systems, where the positions of antenna elements, the allocation of OFDM symbols and subcarriers can be flexibly configured.
		Then, we formulate the problem of estimating target parameters as a tensor decomposition problem admitting to the canonical polyadic format, which enables the parallel target parameters estimation process from corresponding factor matrices along the spatial, temporal, and spectral dimensions, respectively.
		Based on the decomposed factor matrices, we derive the Cramér-Rao Bound (CRB) for the unknown target parameters and reveal that the estimation accuracy of azimuth/elevation angles, velocities and ranges is fundamentally determined by the array geometry, the distribution of OFDM symbols and subcarriers. 
		Building on this insight, we obtain an optimized solution for the positions of antenna elements, and optimal solutions for the subcarrier allocation and OFDM symbol allocation to minimize the CRB, as well as the mean square error of target parameters estimation. 
		Numerical results demonstrate that the proposed tensor decomposition-based wireless sensing with spatial-temporal-spectral optimization can significantly reduce both the CRB and the actual {root mean square error} of wireless sensing compared to conventional MIMO-OFDM ISAC systems employing fixed-position antennas and static time-frequency resource allocation.
		
	\end{abstract}
	
\section{Introduction}

Integrated sensing and communication (ISAC), which serves as a unified platform for fully integrating radar sensing and communication capabilities, has been recognized as a fundamental enabler for sixth-generation (6G) wireless systems \cite{TowardCrowdRenZhihui2025,NineChallengesTongWen2022,EfficientActiveMolaei2025}. 
The realization of ISAC has been facilitated by advanced multicarrier modulation techniques such as orthogonal frequency-division multiplexing (OFDM), which enhances spectral efficiency and robustness against multipath fading, along with other key enabling technologies \cite{ISACSurveyXingzhouWen2022,ASurveyLiuAn2022}. 
Meanwhile, multiple-input multiple-output (MIMO), which utilizes multiple antennas at both the transmitter and receiver to enable spatial multiplexing and diversity, has emerged as one of the most promising technologies for next-generation mobile communication networks \cite{ISACsurveyRuoyuZhang2025,RZhangTerahertzIntegrated2025}. 

Leveraging the spatial degrees of freedom and inherent waveform diversity, the integration of MIMO and ISAC can simultaneously achieve ultra-reliable communication and wireless sensing, thereby facilitating a wide range of advanced applications \cite{MultibeamZhangJ2019,LiPeishiMOI2025,MultiObjectiveWangPeng2025}. 
To achieve this goal, a novel multibeam framework was proposed in \cite{MultibeamZhangJ2019} with a steady subbeam dedicated to communication and a fully steerable one enabling uninterrupted scanning for sensing.
The authors in \cite{10556732} investigated transmit beamforming designs for MIMO ISAC systems by proposing a multiobjective optimization framework to achieve a fair performance tradeoff between multiuser communications and multitarget sensing. Recent ISAC works
developed a joint design methodology that can strike a scalable
tradeoff between the dual functionalities via exploiting the
extra flexibility and the degrees of freedom offered by MIMO
techniques \cite{LiuFanToward2025,LiuJointBeamforming2020}.
However, the aforementioned MIMO ISAC systems generally
configure the fixed-position antennas for both sensing and communication, failing to fully exploit the wireless signal variation in a given transmitter/receiver region and thereby cannot effectively adapt to the ISAC requirements in wireless networks.

In recent years, movable antennas (MAs)/fluid antenna
systems (FAs) have attracted considerable research interest due to their unique capability to flexibly optimize wireless channels through intelligent antenna repositioning \cite{zhuMovableAntennasWireless2024,MovableAntennaNingBoyu2025}. Unlike conventional fixed-position antenna arrays, MA technology introduces an additional dimension of flexibility by enabling physical antenna movement within a certain spatial region, thereby fully exploiting the available spatial degrees of freedom \cite{REMAAChenKangjian2025,ChengzhiYeWCL2026,maMIMOCapacityCharacterization2024}. 
As for wireless communications, the authors in \cite{MAEnhancedZhuLipeng2024} proposed multi-directional descent algorithms to jointly optimize antenna positions and transmit power, {thereby achieving} a substantial reduction in the total transmit power. \cite{GloballyOptimalWuYifei2025} proposed a joint beamforming and MA positioning scheme for minimizing the BS power consumption under user requirements, which is suitable for both ideal and non-ideal channel state information (CSI). The novel algorithm in \cite{GuangyiEnergyEfficiency2025} jointly optimizes beamforming and antenna positioning, which significantly improves the energy efficiency of movable antenna systems. Joint antenna position and rotation optimization was also considered in \cite{Xiaodan6DMA2025}, 
which can fully exploit both the spatial and rotational degrees of freedom of MAs. As for sensing and ISAC, recent works also demonstrate the performance improvement of MAs for pushing
the theoretical limits of sensing via optimizing spatial positioning \cite{PositionBFISAC2025,MAAISAC2025,CRBMinMAA2025,maMovableAntennaEnhanced2024}. 
For instance, the authors in \cite{PositionBFISAC2025} optimize movable antenna positions to minimize the Cramér-Rao Bound (CRB)  for ISAC systems and achieve significant performance gains over fixed antennas. The work in \cite{CRBMinMAA2025} minimizes the CRB for angle estimation in a multiuser ISAC system by jointly optimizing movable antenna positions and beamforming. An MA array-based wireless sensing system was proposed in \cite{maMovableAntennaEnhanced2024}, which significantly {reduces} the CRB and practical mean square error for direction of arrival (DOA) estimation. Although the amalgamation of MA with ISAC has demonstrated significant performance improvement for communication and sensing, the aforementioned works {did ont investigate} how to realize wireless sensing via estimating the information of target parameters in practical systems.

In fact, the target parameter estimation problem
has been extensively studied in the field of conventional radar sensing and recently investigated in ISAC systems\cite{PassiveSensingISAC2025, SensingMUSIC2025,MolaeiISACPE2025}. The authors in \cite{PassiveSensingISAC2025} proposed a robust multi-transmitter passive sensing system for near-field ISAC to enhance sensing performance under imperfect CSI. The study in \cite{SensingMUSIC2025} {proposed} a decoupled multiple signal classification (MUSIC) algorithm for near-field mobile target sensing, which achieves benchmark performance at a drastically reduced computational cost. A high-accuracy parameter estimation method {was} introduced in \cite{MolaeiISACPE2025} for dual-function radar-communication systems using fourth-order statistics and improved matched-filtering.
The work in \cite{SuperResolutionISACOFDM2020} capitalized on the translational invariance structure of OFDM signals across both pulse and frequency domains to achieve auto-paired range and velocity estimation. Additionally, \cite{PerformanceVelocity2021} presented a two-dimensional MUSIC algorithm integrated with a carefully designed smoothing window to mitigate the computational burden. To further improve the sensing estimation accuracy, the authors in \cite{JointRange2022Fuqiang} considered intrapulse and intersubcarrier Doppler effects, leading to a refined maximum likelihood estimation scheme. Building on this, the work in \cite{KeskinJointMIMOOFDM} tackled the joint estimation of delay, Doppler, and angle parameters for multiple targets in MIMO-ISAC systems. Expanding the scope to system architecture, \cite{FrameworkMobile2020} proposed integrating radar sensing into mobile networks and developed direct/indirect schemes for estimating azimuth, range, and velocity via one-dimensional compressive sensing. The authors in \cite{JointTargetZhouLei2025} proposed a joint target detection and multi-user downlink channel estimation method based on sparse Bayesian learning. 

Recently, the tensor-based signal processing has been deemed as a promising solution for wireless sensing, such as DOA estimation, channel estimation, MIMO radar \cite{TensorMIMOradar2025,TensorDMachineLearning2025,Leiyao2025CP,ye2026rotatableantennaenhancedwirelesssensing}, and has  been
successfully applied for ISAC systems \cite{TensorDISAC2025}. 
Different from the aforementioned compressed sensing-based
methods that rely on constructing high-dimensional
dictionary for sparse representation of channels, tensor-based
processing has become one of promising tools for exploiting
the multidimensional characteristics of high-dimensional channel matrices \cite{ChannelEstimationRuoyuZhang2024,RuoyuZhangISACwithMIMO2024,RuoyuZhangChannelTraining2025,TangjunCooperativeISAC2025}. For example, the authors in \cite{RuoyuZhangISACwithMIMO2024} developed a unified tensor model for integrated wireless communication and sensing. In another study, the authors in \cite{RuoyuZhangChannelTraining2025} introduced a tensor-based channel training-assisted sensing framework for terahertz MIMO-ISAC systems. This approach employs tensor decomposition to jointly estimate both channel and target parameters, greatly reducing training overhead while enabling high-precision sensing. Additionally, \cite{TangjunCooperativeISAC2025} proposed a cooperative ISAC scheme for low-altitude drones, which can achieve high-precision estimation of drone parameters in multi-base station scenarios.

From the discussion of aforementioned works, while MA systems can optimize the spatial positions of antennas to enhance wireless sensing performance, and tensor decomposition can improve the estimation accuracy of multiple parameters, there still remain the following two limitations: 1) Conventional {MA-enhanced} ISAC systems only consider the spatial distribution of antennas, which is primarily limited to the spatial-domain optimization and fails to enhance the wireless sensing performance of other parameters beyond angles. 2)  Existing tensor {decomposition-based} algorithms in MIMO-OFDM ISAC {systems operate on} a fixed resource allocation model for sensing and {did not} address the question of how to maximize sensing performance under the general spatial-temporal-spectral resource constraints. To address the limitations above, we develop in this paper a novel wireless sensing framework for MIMO-OFDM ISAC systems with flexible spatial-temporal-spectral optimization and propose a tensor {decomposition-based} approach to estimate target parameters, including azimuth/elevation angles, ranges, and velocities. The specific contributions can be summarized as follows:
	
	\begin{enumerate}
		\item {Firstly, we propose a novel framework for enhanced wireless sensing in MIMO-OFDM ISAC systems, which establishes a monostatic sensing model with flexible spatial-temporal-spectral optimization of antenna positions, OFDM symbols, and subcarriers, where the base station deploys MAs to collect the echo signals and performs target sensing within the limited coherence time and signal bandwidth.  The positions of antenna elements, the allocation of OFDM symbols and subcarriers can be flexibly configured to enhance the target parameter estimation performance.}
		
		\item {We formulate the problem of target sensing as a novel tensor decomposition problem admitting to the canonical polyadic decomposition (CPD) format to accommodate flexible spatial-temporal-spectral configurations. Based on the decomposed factor matrices, we transform the joint estimation of target parameters into parallel estimation processes along the spatial, temporal, and spectral dimensions respectively.}
		
		\item{We derive the CRBs for the unknown target parameters and reveal that the estimation
			accuracy of azimuth/elevation angles, velocities and ranges is fundamentally
			determined by the array geometry, the distribution
			of OFDM symbols and subcarriers. Building on this insight, we
			obtain a globally optimal solution for the positions of antenna
			elements, subcarrier allocation, and OFDM symbol allocation
			to minimize the CRB.}
		
		\item Finally, we conduct sufficient simulation experiments to  present the superiority of our proposed method. The simulation results demonstrate that the performance of the proposed flexible spatial-temporal-spectral optimization via tensor decomposition performs significantly better than the conventional tensor {decomposition-based} method and the conventional subspace-based method, as well as the {root mean square error (RMSE)} of target parameters is closer to CRB.
	\end{enumerate}
	

	Notations: Lowercase letter, boldface lowercase letter, boldface uppercase letter and calligraphy uppercase letter denote the scalars, vectors, matrices and tensors, respectively. The conjugate, transpose, conjugate transpose, pseudo-inverse, mathematical expectation, and inverse are represented by  $(\cdot)^*$, $(\cdot)^{\text{T}}$, $(\cdot)^{\text{H}}$, $(\cdot)^{-1}$  $(\cdot)^{\dagger}$, and $\mathbb{E}\{\cdot\}$. 
	The outer product of vectors is given by $\circ$. The Kronecker product, Khatri-Rao product are denoted by $\otimes$, $\odot$. The expression $\diag(\bm{b})$ denotes the diagonal matrix formed by using vector $\bm{b}$ as its diagonal elements. The expression $ [\bm{a}]_{k}$ and $[\bm{W}]_{p,q}$ denotes the $k$-th element of the vector $\bm{a}$ and the $i$-th row and the $j$-th column of the matrix $\bm{W}$, respectively. 
	The expression $\var({\mathcal{X}})$ denotes the variance of the set $\mathcal{X}$ and $\var(\mathcal{X}) = \frac{1}{N}\sum_{n=1}^{N}x_n^2 - ( \frac{1}{N}\sum_{n=1}^{N}x_n )^2$. The expression $\cov(\mathcal{X},\mathcal{Y})$ denotes the covariance of the set $\mathcal{X},\mathcal{Y}$, and $\cov(\mathcal{X},\mathcal{Y}) = \frac{1}{N}\sum_{n=1}^{N} [ ( x_n - \frac{1}{N}\sum_{n=1}^{N}x_n ) ( y_n - \frac{1}{N}\sum_{n=1}^{N}y_n ) ]$. The set of positive integers is denoted by $\mathbb{Z}^+$. The real part and the imaginary part of the complex matrix or complex number are expressed by $\mathfrak{R}[\cdot]$ and $\mathfrak{I}[\cdot]$, respectively. 
	The $\ell_2 $ norm of a vector is denoted by $\|\cdot\|_2$. 
	The Frobenius norm of a matrix is expressed as $\|\cdot\|_F$.
	The flooring of positive integer $m$ is denoted by $\lfloor m \rfloor$.
	
	\section{System model}
	In this section, as shown in \Cref{Fig_BSUE}, we present the system model for communication and radar sensing in MIMO-OFDM ISAC systems.
	\subsection{System Description}
	As shown in \Cref{Fig_BSUE}, the transmit part of the base station (BS) is equipped with a uniform rectangular array (URA) with $N_{\text{bs},x}$ and $N_{\text{bs},y}$ fixed-position antennas located in the $x$ and $y$ directions, such that $N=N_{\text{bs},x} N_{\text{bs},y}$. The receive part of the BS is equipped with $N_{{\text{re}}}$ MAs, whose coordinates are expressed as $\mathcal{X}=\{x_1,x_2,\ldots,x_{N_{{\text{re}}}}\}$ and $\mathcal{Y}=\{y_1,y_2,\ldots,y_{N_{{\text{re}}}}\}$. We consider the MAs can be flexibly adjusted within the two-dimensional movable region $\mathcal{C}$, which is assumed to be a rectangular region of size $A_x\times A_y$. We utilize OFDM-based ISAC signals to realize communication and sensing, which consist of $K_0$ total subcarriers and $M_0$ total OFDM symbols, whose sets are respectively given by ${\mathcal{F}_0=\{1,\ldots,K_0}\}$ and ${\mathcal{T}_0=\{1,\ldots,M_0}\}$. For wireless sensing, we select a set $\mathcal{T}$ of $M$ OFDM symbols from $\mathcal{T}_0$ and a set $\mathcal{F}$ of $K$ subcarriers from $\mathcal{F}_0$, such that $\mathcal{T}\subseteq \mathcal{T}_{0}$ and $\mathcal{F}\subseteq \mathcal{F}_{0}$. In this paper, we focus on simultaneously sensing $Q$ targets by exploiting spatial-temporal-spectral resources, whose distribution directly influences the sensing performance. It should be noted that the positions of the MAs, as well as the distribution of OFDM symbols and subcarriers, can be flexibly adjusted. Thus, we perform the flexible spatial-temporal-spectral optimization for $\mathcal{X},\mathcal{Y},\mathcal{F}$, and $\mathcal{T}$ jointly.
	
	\subsection{Signal Model}
	\begin{figure}
		\centering
		\includegraphics[width=3in]
		{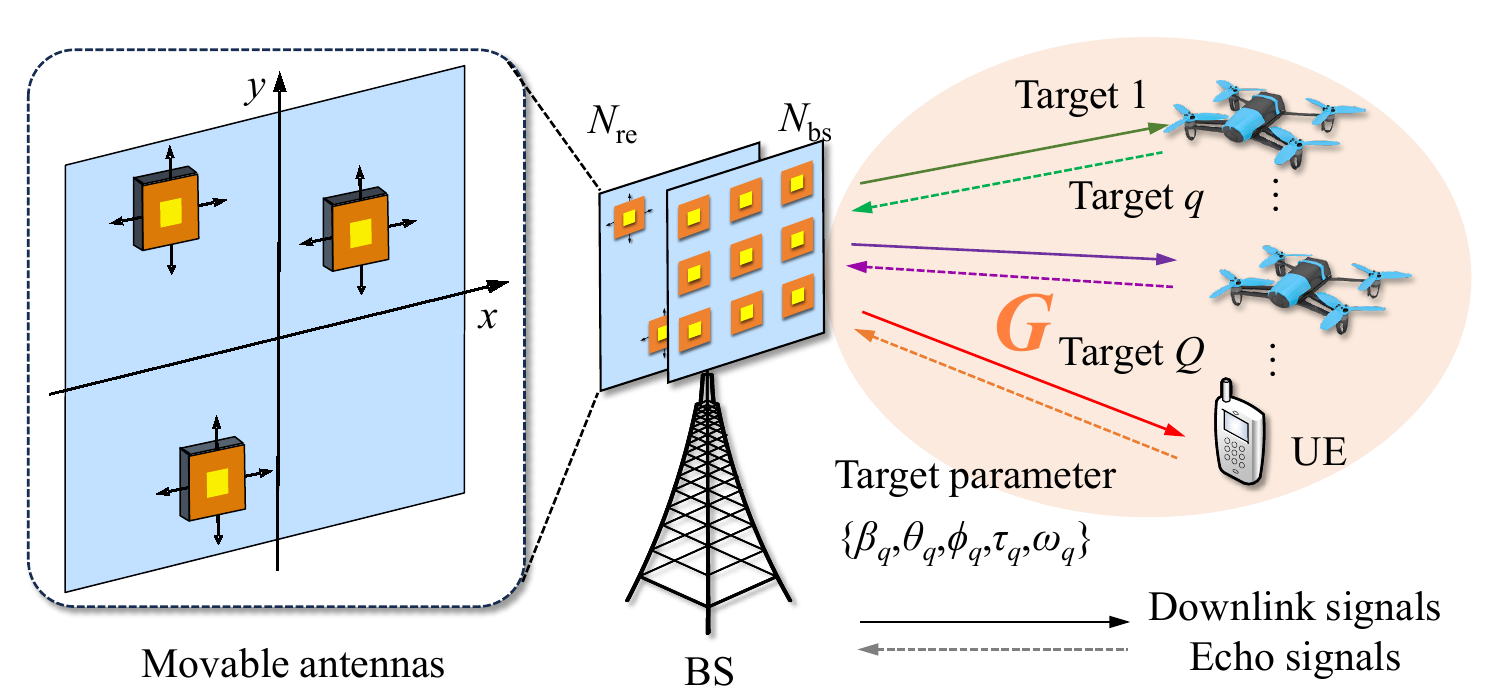}
		\caption{An illustrative MIMO-OFDM ISAC system model, where the BS is equipped with MAs.}
		\label{Fig_BSUE}
		\vspace{-0.5em}
	\end{figure}
	In this subsection, the signal model for wireless sensing in  MIMO-OFDM ISAC is presented. The complex baseband time-domain transmit signal corresponding to the $m$-th OFDM symbol for the $q$-th target is given by
\begin{align}
		\!\tilde{s}_{m}(t)\!=\!\frac{1}{\sqrt {K_0}}\sum_{k=1}^{K_0}\!\chi_{m,k}\text{e}^{\text{j}2\pi (k-1)\Delta f t}\rect\big(\frac{t-m{T}_{\text{sym}}}{T_{\text{sym}}}\big),
	\end{align}
	where $\chi_{m,k}\in \mathbb{C}$ denotes the transmitted payload data symbol at the $k$-th subcarrier of the $m$-th OFDM symbol and $|\chi_{m,k}|^{2}=1$. $\Delta f=f_{\text{s}}/{K_0}$ denotes the subcarrier spacing. $\rect(t)$ denotes the transmit pulse shaping filter. $T_{\text{sym}}=T_{\text{eff}}+T_{\text{cp}}$ represents the duration of an effective OFDM symbol and $T_{\text{cp}}$ is the cyclic-prefix duration. Then, the transmitted frequency domain signal vector can be expressed as
{	\begin{align}
		\bm{s}_{m,k}=\bm{p}_m\chi_{m,k},
	\end{align}}
	where $\bm{p}_{m}\in \mathbb{C}^{N_{\text{bs}}\times 1}$ denotes the transmit beamforming vector. After removing the cyclic-prefix and performing discrete Fourier transform at the base station, the echo signal can be expressed as \cite{RuoyuZhangISACwithMIMO2024}
	\begin{align}
		\tilde{\bm{z}}_{m,k}=\bm{G}_{m,k}{\bm{s}}_{m,k}+\tilde{\bm{n}}_{m,k}\label{modely},
	\end{align}
	where $\bm{G}_{m,k}$ denotes the sensing channel matrix at the $k$-th subcarrier of and $m$-th OFDM symbols. $\tilde{\bm{n}}_{m,k}\sim\mathcal{CN}\big( \bm{0}_{N_{\text{\text{re}}}},\sigma_{n}^{2}\bm{I}_{N_{\text{\text{re}}}\times N_{\text{\text{re}}}} \big)$ denotes additive white Gaussian noise, and $\sigma_{n}^{2}$ denotes the power of noise. It should be noted that the echo signal in \eqref{modely}, composed of multiple OFDM symbols and subcarriers, cannot be naturally decomposed into separate dimensions, thus not inherently conforming to typical tensor models \cite{Cheng2017Probabilistic}. 
	 After removing the pilot symbols at the receiver, the received signal can be further expressed as 
	\begin{align}
		\bm{z}_{m,k} &= \chi_{m,k}^{*}\tilde{\bm{z}}_{m,k} = \bm{G}_{m,k}{\bm{p}}_{m}+{\bm{n}}_{m,k}  \label{ymk},
	\end{align}
	where $\bm{n}_{m,k}=\chi_{m}^{*}\tilde{\bm{n}}_{m,k}$ denotes the corresponding noise vector. In the next subsection, we will demonstrate the sensing channel in particular.
	
	\subsection{Sensing Channel Model}
	Given that all signals are transmitted and received from the same BS, their angles of arrival and departure are geometrically identical for the $Q$ mutually uncorrelated targets assumed to be located in the far-field region. The time and delay domain sensing channel can be expressed as
\begin{align}
	\!\bm{G}(t,\tau)=\sum\limits_{q=1}^{Q}{{\beta}_{q}\varrho(\tau-\tau_q){\bm{a}}_{\text{\text{re}}}({\vartheta}_{q},\!{\varphi}_{q})}\bm{a}_{\text{bs}}^{\text{T}}({\vartheta}_{q},\!{\varphi}_{q}){\text{e}^{\text{j}2\uppi {\omega}_{q}t}},
\end{align}
	where $\varrho(t)$ denotes the impulse function. The complex reflection coefficient ${\beta}_{q}$ represents the radar cross section of the $q$-th target. The round-trip time delay ${\tau}_{q}=2{R}_{q}/c$ and Doppler shift ${\omega}_{q}={f}_{c}(2{V}_{q}/c)$ are determined by the distance ${R}_{q}$ and velocity ${V}_{q}$ of the $q$-th target, where $c$ denotes the speed of the light.
	${\vartheta}_{q}=\sin {\theta}_{q}\cos {\phi}_{q}$, ${\varphi}_{q}=\cos {\theta}_{q}$, where ${\theta}_{q}$ and ${\phi}_{q}$ denotes the elevation and azimuth angles. The receiver steering vector ${\bm{a}}_{\text{\text{re}}}({\vartheta}_{q},{\varphi}_{q})\in \mathbb{C}^{{N}_{\text{\text{re}}}\times 1}$ is given by
	\begin{align}
		\!\!\!\!	{\bm{a}}_{\text{\text{re}}}({\vartheta}_{q},{\varphi}_{q})\!=\!\big[{\text{e}^{-\text{j}\frac{2\uppi}{{\lambda}_{c}}(x_1 \vartheta_q+y_1 \varphi_q) }},\ldots,{\text{e}^{\!-\text{j}\frac{2\uppi}{{\lambda}_{c}}\!(\!x_{N_{\text{re}}} \vartheta_q+y_{N_{\text{re}}} \varphi_q\!)}} \big]^{\text{T}}\!,
	\end{align}
	where ${\lambda}_{c}$ denotes the carrier wavelength.
	The steering vector at {the} BS with URA configuration is ${\bm{a}}_{\text{bs}}({\vartheta}_{q},{\varphi}_{q})={\bm{a}}_{\text{bs},x}({\vartheta}_{q},{\varphi}_{q})\otimes {\bm{a}}_{\text{bs},y}({\vartheta}_{q},{\varphi}_{q})\in \mathbb{C}^{{N}_{\text{bs},x}{N}_{\text{bs},y}\times 1}$, where $x$ and $y$ direction steering vectors are expressed as
	\begin{align}
		{\bm{a}}_{\text{bs},x}({\vartheta}_{q},{\varphi}_{q})&=\big[ 1,\ldots,{\text{e}^{-\text{j}\frac{2\uppi}{{\lambda}_{c}}({N}_{\text{bs},x}-1)d {\vartheta}_{q} }} \big]^{\text{T}} ,\\
		{\bm{a}}_{\text{bs},y}({\vartheta}_{q},{\phi}_{q})&=\big[ 1,\ldots,{\text{e}^{-\text{j}\frac{2\uppi}{{\lambda}_{c}}({N}_{\text{bs},y}-1)d {\varphi}_{q}}} \big]^{\text{T}},
	\end{align}
	where $d=\lambda_c/2$ represents the spacing between adjacent antenna elements.
	It should be noted that in wireless communication, the reflection coefficients in the propagation environment correspond to the target reflection strength in wireless sensing. 
	The sensing channel matrix at the $k$-th subcarrier and $m$-th OFDM symbols can be expressed as 
	\begin{align}
		\!\!\!	\bm{G}_{m\!,k}\!=\!\sum_{q=1}^{Q}\!\beta_{q}\bm{a}_{\text{\text{re}}}(\vartheta_{q},\!\varphi_{q})\bm{a}_{\text{bs}}^{\text{T}}(\vartheta_{q},\!\varphi_{q}) {\text{e}^{-\text{j}2\uppi \tau_{q}{k}\Delta f{+{\text{j}2\uppi m\omega_{q}T_{\text{sym}}}}}}.\label{Gmk}
	\end{align}

Our goal is to achieve the optimized performances of estimating the parameters by flexibly adjusting $\mathcal{X},\mathcal{Y},\mathcal{T}$, and $\mathcal{F}$ within the limited resources. To achieve this goal, we propose the tensor decomposition-based algorithm and the flexible spatial-temporal-spectral optimization in the following.

	\section{Wireless Sensing for MIMO-OFDM ISAC-Based on Tensor Decomposition}
	
	In this section, we propose a tensor-based parameter estimation algorithm, which is suitable for {a variety of parameter estimation applications}. Specifically, we first model the received signal as a tensor and perform the CPD using the alternating least squares (ALS) algorithm. Then, we exploit the spectrum peak search method to estimation the parameters.
	\subsection{Tensor Formulation and Factor Matrices Estimation}\label{section II}
	For wireless sensing, the positions of MAs, the OFDM symbols and the subcarriers can be flexibly adjusted or selected. Specifically, the positions of MAs satisfy $\mathcal{X},\mathcal{Y}\in\mathcal{C}$, the number of distinct selections for $\mathcal{F}$ and $\mathcal{T}$ is $\binom{M_0}{M}$ and $\binom{K_0}{K}$, respectively.
	Considering $M$ OFDM symbols, \eqref{ymk} can be rewritten as\footnote{In this manuscript, we follow the common assumption in existing ISAC literature \cite{RuoyuZhangISACwithMIMO2024,RuoyuZhangChannelTraining2025,TangjunCooperativeISAC2025} that the self interference is perfectly cancelled.}
	\begin{align}
		\bm{Z}_{k}=\sum\limits_{q=1}^{Q}&\beta_{q}\bm{a}_{\text{\text{re}}}(\mathcal{X},\mathcal{Y},\vartheta_{q},\varphi_{q})\bm{a}_{\text{bs}}^{\text{T}}(\vartheta_{q},\varphi_{q})\bm{P\varGamma}(\mathcal{T},\omega_{q})\nonumber\\
		&\times{\text{e}^{-\text{j}2\uppi \tau_{q}\mathcal{F}(k)\Delta f}}+\bm{N}_{k},\label{Yk}
	\end{align}
	where $\bm{Z}_{k}=[\bm{z}_{1,k},\bm{z}_{2,k},\ldots,\bm{z}_{M,k}]\in \mathbb{C}^{N_{\text{\text{re}}}\times M}$ represents the received signal at the $k$-th subcarrier. The pilot training matrix $\bm{P}=[\bm{p}_{1},\ldots,\bm{p}_{M}]\in \mathbb{C}^{N_{\text{\text{bs}}}\times M}$, and $\bm{\varGamma}(\mathcal{T},\omega_{q})\in \mathbb{C}^{M\times M}=\diag[\bm{a}_{\text{do}}(\mathcal{T},\omega_{q})]$, where $\bm{a}_{\text{do}}(\mathcal{T},\omega_{q})$ represents the equivalent Doppler steering vector and is given by 
	\begin{align}
		\bm{a}_{\text{do}}(\mathcal{T},\omega_{q}) = \big[ \text{e}^{\text{j}2\uppi T_{\text{sym}}\mathcal{T}(1)\omega_{q}}, \ldots, \text{e}^{\text{j}2\uppi T_{\text{sym}}\mathcal{T}(M)\omega_{q}} \big]^{\text{T}}.
	\end{align}
	
	Simultaneously, the training signals received at UE are utilized for channel estimation. In the following, we model the received signal as a tensor and employ the ALS method to perform the CPD on the tensor, we then estimate the parameters from estimated factor matrices.

	Recalling the expression of the received signal and the channel model provided in \Cref{section II}, we {notice} that the AoAs, AoDs, Doppler shifts, and time delays caused by the targets independently affect the received signal in spatial-temporal-spectral dimensions.  Thanks to the characteristic of tensor decomposition that can decouple various parameters, we formulate the multi-dimensional parameter estimation via using a tensor decomposition model. Specifically,
	by combining the reflected signals from $K$ subcarriers, we reconstruct the signal as a third-order tensor $\bm{\mathcal{Z}}\in \mathbb{C}^{N_{\text{\text{re}}}\times M\times K}$ which is given by 
	\begin{align}
		\bm{\mathcal{Z}} =&\sum\limits_{q=1}^{Q} \bm{a}_{\text{\text{re}}}(\mathcal{X},\mathcal{Y},\vartheta_{q},\varphi_{q}) \circ \bm{b}_{\text{bs}}(\mathcal{T},\omega_{q}, \vartheta_{q},\varphi_{q}) \circ \bm{c}_{\text{td}}(\mathcal{F},\beta_{q},\tau_{q})\nonumber \\&+ \bm{\mathcal{N}},\label{Ytensor}
	\end{align}
	where $\bm{\mathcal{N}} \in \mathbb{C}^{N_{\text{\text{re}}} \times M \times K}$ denotes the noise tensor. $\bm{b}_{\text{bs}}(\mathcal{T},\omega_{q},\vartheta_{q},\varphi_{q}) = \bm{\varGamma}(\mathcal{T},\omega_{q})\bm{P}^{\text{T}}\bm{a}_{\text{bs}}(\vartheta_{q},\varphi_{q})$ represents the equivalent transmit steering vector. The equivalent delay steering vector is denoted as $\bm{c}_{\text{td}}(\mathcal{F},\beta_{q},\tau_{q}) = \beta_{q}\bm{a}_{\text{td}}(\mathcal{F},\tau_{q})$, where $\bm{a}_{\text{td}}(\mathcal{F},\tau_{q})$ is given by
	\begin{align}
		\bm{a}_{\text{td}}(\mathcal{F},\tau_{q}) = \big[ \text{e}^{-\text{j}2\uppi \tau_{q}\Delta f\mathcal{F}(1)}, \ldots, \text{e}^{-\text{j}2\uppi \tau_{q}\Delta f\mathcal{F}(K)} \big]^{\text{T}}.
	\end{align}
	Our objective is to estimate the unknown target parameters $\{ \vartheta_{q},\varphi_{q},v_{q},\tau_{q},\beta_{q} \}_{q=1}^{Q}$. This sensing problem can be formulated as \eqref{TensorY}.
		\begin{figure*}
			{\begin{align}
				\min_{\substack{
						\{\mathcal{X},\mathcal{Y}, \mathcal{T}, \mathcal{F}, \vartheta_{q}, \varphi_{q}, \omega_{q}
						\tau_{q}, \beta_{q} \}_{q=1}^{Q}
				}} 
				\bigg\| &\bm{\mathcal{Z}} - \sum_{q=1}^{Q} \bm{a}_{\text{\text{re}}}(\mathcal{X},\mathcal{Y}, \vartheta_{q}, \varphi_{q}) 
				\circ \bm{b}_{\text{bs}}(\mathcal{T},\omega_{q},\vartheta_{q}, \varphi_{q}) 
				\circ \bm{c}_{\text{td}}(\mathcal{F},\beta_{q},\tau_{q}) \bigg\|_{F}^{2},\label{TensorY}
		\end{align}}
	\vspace{-0.5em}
		\hrule
		\vspace{-0.75em}
	\end{figure*}
	
	{Following} \cite{RuoyuZhangISACwithMIMO2024}, the CPD of the received signal tensor in \eqref{Ytensor} can be given by
	\begin{align}
		\bm{\mathcal{Z}} =\llbracket\bm{C}^{(1)},\bm{C}^{(2)},\bm{C}^{(3)}\rrbracket+ \bm{\mathcal{N}}\label{YCP},
	\end{align}
	where the three factor matrices $\bm{C}^{(1)}\in  \mathbb{C}^{N_{\text{\text{re}}} \times Q}$, $\bm{C}^{(2)}\in \mathbb{C}^{M \times Q}$, and $\bm{C}^{(3)}\in  \mathbb{C}^{K \times Q}$ are respectively expressed as
	\begin{align}
		&\bm{C}^{(1)} \!= \! \big[ \bm{a}_{\text{\text{re}}}(\mathcal{X},\mathcal{Y},\vartheta_{1},\varphi_{1}), \ldots, \bm{a}_{\text{\text{re}}}(\mathcal{X},\mathcal{Y},\vartheta_{q},\varphi_{q}) \big]  \label{AA},\\
		&\bm{C}^{(2)} \! = \! \big[ \bm{b}_{\text{bs}}(\mathcal{T},\! \omega_{1}, \vartheta_{1},\! \varphi_{1}), \ldots, \bm{b}_{\text{bs}}( \mathcal{T},\! \omega_{q}, \vartheta_{q}, \varphi_{q}) \big] \label{BB} ,\\
		&\bm{C}^{(3)} \!= \! \big[ \bm{c}_{\text{td}}(\mathcal{F},\beta_{q},\tau_{q}), \ldots, \bm{c}_{\text{td}}(\mathcal{F},\beta_{q},\tau_{q}) \big] \label{CC}.
	\end{align}
	Various methods have been proposed to solve the CPD problem and acquire the factor matrices. One well-known approach is the ALS method, which updates one factor matrix while fixing the others until convergence \cite{Zhang2022Tensor}. For instance, the CPD of tensor $\bm{\mathcal{Z}}$ can be given as
	\begin{align}
		\bm{C}_{(\iota+1)}^{(1)} &= \arg \underset{\bm{C}_{(\iota)}^{(1)}}{\min} \big\| \bm{Z}_{(1)} - \bm{C}^{(1)} \big( \bm{C}_{(\iota)}^{(3)} \odot \bm{C}_{(\iota)}^{(2)} \big)^{\text{T}} \big\|_{F}^{2} \label{C1_k+1},\\
		\bm{C}_{(\iota+1)}^{(2)} &= \arg \underset{\bm{C}_{(\iota)}^{(2)}}{\min} \big\| \bm{Z}_{(2)} - \bm{C}^{(2)} \big( \bm{C}_{(\iota)}^{(3)} \odot \bm{C}_{(\iota+1)}^{(1)} \big)^{\text{T}} \big\|_{F}^{2}\label{C2_k+1}, \\
		\bm{C}_{(\iota+1)}^{(3)} &= \arg \underset{\bm{C}_{(\iota)}^{(3)}}{\min} \big\| \bm{Z}_{(3)} - \bm{C}^{(3)} \big( \bm{C}_{(\iota+1)}^{(1)} \odot \bm{C}_{(\iota+1)}^{(2)} \big)^{\text{T}} \big\|_{F}^{2}\label{C3_k+1},
	\end{align}
	where $\bm{Z}_{(1)}$, $\bm{Z}_{(2)}$, and $\bm{Z}_{(3)}$ represent the mode-1, mode-2, and mode-3 unfoldings of $\bm{\mathcal{Z}}$ respectively. Through iterative execution of these three update steps, we obtain the estimated factor matrices $\big\{ \hat{\bm{C}}^{(1)},\hat{\bm{C}}^{(2)},\hat{\bm{C}}^{(3)} \big\}$ upon convergence. It should be noted that the factor matrices obtained via CPD typically exhibit biases and are subject to two types of ambiguities  scaling ambiguity and permutation ambiguity \cite{kolda2009tensor}. For the estimated factor matrices $\big\{ \hat{\bm{C}}^{(1)},\hat{\bm{C}}^{(2)},\hat{\bm{C}}^{(3)} \big\}$, they satisfy the following relationships 
	\begin{align}
		\hat{\bm{C}}^{(1)} &\triangleq \bm{C}^{(1)}\bm{{\varPi}}\bm{{\varDelta}}^{(1)} + \bm{E}^{(1)} \label{AAhat},\\
		\hat{\bm{C}}^{(2)} &\triangleq \bm{C}^{(2)}\bm{{\varPi}}\bm{{\varDelta}}^{(2)} + \bm{E}^{(2)} \label{BBhat},\\
		\hat{\bm{C}}^{(3)} &\triangleq \bm{C}^{(3)}\bm{{\varPi}}\bm{{\varDelta}}^{(3)} + \bm{E}^{(3)}\label{CChat},
	\end{align}
	where $\bm{\varPi} \in \mathbb{C}^{Q \times Q}$ is a permutation matrix, and $\bm{\varDelta}^{(1)}$, $\bm{\varDelta}^{(2)}$, $\bm{\varDelta}^{(3)}$ are diagonal matrices. The matrices $\bm{E}^{(1)}$, $\bm{E}^{(2)}$, and $\bm{E}^{(3)}$ represent the estimation error terms.
	From \eqref{AAhat}-\eqref{CChat}, we observe that the estimated factor matrices differ from the original ones, even in noise-free scenarios. Fortunately, these ambiguities do not disrupt the pairing relationships between the estimated and true factor matrices, enabling the extraction of AoA information from corresponding column vectors.
	In the next subsection, we will propose the spectrum peak search algorithm to estimate the unknown parameters.

	\subsection{Parameter Estimation} \label{section IV}
	In Section \ref{section II}, we employed the ALS algorithm to perform CPD of the tensor and estimate the factor matrices. In this subsection, we proceed to estimate the elevation and azimuth angles, time delays, Doppler shifts, and reflection coefficients of targets from these estimated factor matrices. 
	
	In \eqref{AA}, we observe that each column of the factor matrix $\hat{\bm{C}}^{(1)}$ corresponds to a pair of elevation and azimuth angles. It should be noted that the permutation matrix $\bm{\varPi}$ remains consistent across all factor matrices, thus its presence does not affect parameter estimation. Therefore, we can estimate the $\vartheta$ and $\varphi$ associated with each column of $\hat{\bm{C}}^{(1)}$ using a correlation-based method 
	\begin{align}\label{hatthetaphi}
		\!\!\big\{ \hat{\vartheta}_q, \hat{\varphi}_q \big\} =&\argmin_{\vartheta_q,\phi_q}\|\hat{\bm{c}}_q^{(1)}- \bm{a}_{\text{\text{re}}}(\mathcal{X},\mathcal{Y},\vartheta_q,\varphi_q)\|_2^2\\
		=&\argmax_{\vartheta_q,\varphi_q} \frac{\big| (\hat{\bm{c}}_q^{(1)})^{\text{H}} \bm{a}_{\text{\text{re}}}(\mathcal{X},\mathcal{Y},\vartheta_q,\varphi_q) \big|^2}{\|\hat{\bm{c}}_q^{(1)}\|_2^2 \|\bm{a}_{\text{\text{re}}}(\mathcal{X},\mathcal{Y},\vartheta_q,\varphi_q)\|^2_2},\nonumber
	\end{align}
	where $\hat{\bm{c}}_q^{(1)}$ denotes the $q$-th column of the estimated factor matrix $\hat{\bm{C}}^{(1)}$. Then the elevation and azimuth angles can be estimated by $\hat\theta=\arccos\hat{\phi}$ and $\hat\phi=\arccos {\hat\vartheta}/ {\sqrt{1-\hat{\varphi}^2}}$.
	
	The $q$-th column of \eqref{CChat} satisfies 
	\begin{align}
		\hat{\bm{c}}_{q}^{(2)} &= \delta_{q}^{(2)}\bm{\varGamma}(\mathcal{T},\omega_{q})\bm{P}^{\text{T}}\bm{a}_{\text{bs}}(\mathcal{X},\mathcal{Y},\theta_{q},\phi_{q}) + \bm{e}_{q}^{(2)} \nonumber
		\\ &= \delta_{q}^{(2)}\bm{\hat{Q}}^{\text{T}}\bm{a}_{\text{do}}(\mathcal{T},\omega_q) + \bm{e}_{q}^{(2)},\label{cchatchange}
	\end{align}
	where the second equality holds due to the definition of diagonal matrix $\bm{\hat Q} \in \mathbb{C}^{M \times M}$ with diagonal elements $\bm{P}^{\text{T}}\bm{a}_{\text{bs}}(\hat{\theta}_{q},\hat{\phi}_{q})$, and $\delta_{q}^{(2)}$ denotes the $(q,q)$-th element of $\bm{\varDelta}^{(2)}$. Since the elevation and azimuth angles of AoA/AoD are equal, the AoA can be estimated from \eqref{hatthetaphi}, and with the known precoding matrix $\bm{P}$, \eqref{cchatchange} can be rewritten as
	\begin{align}\label{tildebbhat}
		\tilde{\hat{\bm{c}}}_{q}^{(2)} = \delta_{q}^{(2)}\bm{a}_{\text{do}}(\mathcal{T},\omega_q)  + \tilde{\bm{e}}_{q}^{(2)},
	\end{align}
	where $\tilde{\hat{\bm{c}}}_{q}^{(2)} = (\bm{\hat{Q}}^{\text{T}})^{-1}\hat{\bm{c}}_{q}^{(2)}$, $\tilde{\bm{e}}_{q}^{(2)} = (\bm{\hat{Q}}^{\text{T}})^{-1}\bm{e}_{q}^{(2)}$, combining all $Q$ targets, we obtain the matrix form of \eqref{tildebbhat} 
	\begin{align}\label{tildeBBhat}
		\tilde{\hat{\bm{C}}}^{(2)} = \bm{A}_{\text{do}}\bm{\varDelta}^{(2)} + \tilde{\bm{E}}^{(2)},
	\end{align}
	where $\bm{A}_{\text{do}}=[\bm{a}_{\text{do}}(\mathcal{T},\omega_1),\ldots,\bm{a}_{\text{do}}(\mathcal{T},\omega_Q)]$.
	We now turn to estimating the Doppler shifts from the factor matrix $\tilde{\hat{\bm{C}}}^{(2)}$. Substituting the estimated $\big\{ \hat{\vartheta}_q, \hat{\varphi}_q \big\}$ into \eqref{tildebbhat},\footnote{For multistatic sensing systems where both AoD and Doppler shift are unknown in $\hat{\bm{C}}^{(2)}$, the alternating iterative algorithm in \cite{RuoyuZhangChannelTraining2025} can be employed to estimate AoD and Doppler shift simultaneously.} Similarly, we can estimate $\omega_q$ by computing
	\begin{align}\label{hatDop}
		\hat{\omega}_q 
		&= \argmax_{\omega_q} \frac{\big| (\tilde{\hat{\bm{c}}}_q^{(2)})^{\text{H}} \bm{a}_{\text{do}}(\mathcal{T},\omega_q) \big|^2}{\|\tilde{\hat{\bm{c}}}_q^{(2)}\|_2^2 \|\bm{a}_{\text{do}}(\mathcal{T},\omega_q)\|_2^2}.
	\end{align}
	Also, based on \eqref{CChat}, we can proceed to estimate $\tau_q$ via
	\begin{align}
		\hat{\tau}_q 
		&= \argmax_{\tau_q} \frac{\big| (\hat{\bm{c}}_q^{(3)})^{\text{H}} \bm{a}_{\text{td}}(\mathcal{F},\tau_q) \big|^2}{\|\hat{\bm{c}}_q^{(3)}\|_2^2 \|\bm{a}_{\text{td}}(\mathcal{F},\tau_q)\|_2^2},\label{hattau}
	\end{align}
	where $\hat{\bm{c}}_q^{(3)}$ represents the $q$-th column of $\hat{\bm{C}}^{(3)}$.

	Following this estimation procedure, after obtaining $\{ \hat{\theta}_q, \hat{\phi}_q, \hat{\omega}_q, \hat{\tau}_q \}_{q=1}^Q$, we proceed to estimate the reflection coefficients $\{ \beta_q \}_{q=1}^Q$. Based on the CPD in \eqref{CC}, we unfold $\bm{\mathcal{Z}}$ along mode-3, transpose it, and vectorize the columns to obtain 
	\begin{align}
		\text{vec} (\bm{Z}_{(3)}^{\text{T}} ) = \bm{A}_{\text{td}} \odot (\bm{C}^{(2)} \odot \bm{C}^{(1)}) \bm{\beta} + \text{vec}( \bm{N}_{(3)}^{\text{T}} ),
	\end{align}
	where $\bm{A}_{\text{td}} = \big[ \bm{a}_{\text{td}}(\mathcal{F},\tau_1), \ldots, \bm{a}_{\text{td}}(\mathcal{F},\tau_Q) \big]$ represents the delay steering matrix. Therefore, we estimate $\hat{\bm{\beta}}$ as
	\begin{align}\label{hatbeta}
		\hat{\bm{\beta}} = [ \hat{\bm{A}}_{\text{td}} \odot (\bar{\bm{C}}^{(2)} \odot \bar{\bm{C}}^{(1)}) ]^\dagger \text{vec}( \bm{Z}_{(3)}^{\text{T}} ),
	\end{align}
	where $\bar{\bm{C}}^{(1)}=[\bm{b}_{\text{bs}}(\mathcal{X},\mathcal{Y},\hat{\vartheta}_1,\hat{\varphi}_1),\ldots,\bm{b}_{\text{bs}}(\mathcal{X},\mathcal{Y},\hat{\vartheta}_Q,\hat{\varphi}_Q)]$,  $\bar{\bm{C}}^{(2)}=[\bm{b}_{\text{bs}}(\mathcal{T},\hat{\omega}_1,\hat{\vartheta}_1,\hat{\varphi}_1),\ldots,\bm{b}_{\text{bs}}(\mathcal{T},\hat{\omega}_Q,\hat{\vartheta}_Q,\hat{\varphi}_Q)]$, and $\hat{\bm{A}}_{\text{td}} = \big[ \bm{a}_{\text{td}}(\mathcal{F},\hat{\tau}_1), \ldots, \bm{a}_{\text{td}}(\mathcal{F},\hat{\tau}_Q) \big]$.
	It is important to note that the estimation of reflection coefficients depends on $\{ \hat{\vartheta}_q, \hat{\varphi}_q, \hat{\omega}_q, \hat{\tau}_q \big\}_{q=1}^Q$. Since the optimization in \eqref{maxK}-\eqref{maxM} improves the CRB for $\vartheta$, $\varphi$, $\omega$, and $\tau$, the estimation performance of $\hat{\bm{\beta}}$ is consequently enhanced.

	Benefiting from the ability of CPD to decouple the parameters from multidimensional signals into multiple independent groups, this approach reduces the dimensionality of the Fisher information matrix, thereby enabling the derivation of specific CRB expressions for individual parameters and laying the foundation for flexible spatial-temporal-spectral optimization. In the next subsection, we will present the derivation of CRB for each parameter individually.
	
	Next, we analyze the complexity of the proposed algorithm. 
	For a received signal tensor of dimensions $N_{\text{re}} \times K \times M$ and rank $Q$, the complexity per iteration for updating the factor matrices is approximately $\mathcal{O}(N_{\text{re}} K M Q)$. Given that the subsequent parameter extraction involves only simple phase operations with a negligible complexity of $\mathcal{O}(Q N_{\text{re}} + Q K + Q M)$, the total complexity of the algorithm after $I_{\text{iter}}$ iterations is effectively$ \mathcal{O}(I_{\text{iter}} N_{\text{re}} K M Q)$, where $N_{\text{re}}$, $K$, and $M$ denote the number of receive antennas, subcarriers, and OFDM symbols.
	
	\section{Flexible
		Spatial-Temporal-Spectral Optimization for Wireless Sensing}\label{section IIV}
	
	In this section, we first derive the CRB of each parameter based on the estimated factor matrices. Then, we transform the problem of minimizing the CRB into the flexible spatial-temporal-spectral optimization with respect to $\mathcal{X},\mathcal{Y},\mathcal{T}$, and $\mathcal{F}$.
	\subsection{Derivation of CRB Based on Estimated Factor Matrices}

	In this subsection, {we first} illustrate the difficulty of deriving the specific expression of CRB for each parameter from the received signal matrix. {Then}, we derive the CRB of each parameter individually based on the expression of the estimated factor matrices.
	
	Our goal is to improve the performance of estimation, which can be transformed into minimizing the CRB of each parameter.
	As is well known, the normalized root mean square error of parameter estimation represents the performance of the parameter estimation. For convenience of notation, let $\bm{\mu} = [ \bm{\beta}^{\text{T}}, \bm{\vartheta}^{\text{T}}, \bm{\varphi}^{\text{T}}, \bm{\omega}^{\text{T}}, \bm{\tau}^{\text{T}} ] \in \mathbb{C}^{5Q \times 1}$, where $\bm{\beta} = {[ \beta_{1}, \ldots, \beta_{Q} ]}^{\text{T}}$, $\bm{\vartheta} = {[ \vartheta_{1}, \ldots, \vartheta_{Q} ]}^{\text{T}}$, $\bm{\varphi} = {[ \varphi_{1}, \ldots, \varphi_{Q} ]}^{\text{T}}$, $\bm{\omega} = {[ \omega_{1}, \ldots, \omega_{Q} ]}^{\text{T}}$, and $\bm{\tau} = {[ \tau_{1}, \ldots, \tau_{Q} ]}^{\text{T}}$. The RMSE of the $q$-th parameter ${[\bm{\eta}]}_{q}$ and the CRB of this parameter satisfy the following relationship 
{	\begin{align}
		\text{RMSE}({[{\bm{\eta }}]_q}) \ge \sqrt {CRB({{[{\bm{\eta }}]}_q})} .
	\end{align}}
	The typical CRB is derived based on \eqref{C1_k+1}-\eqref{C2_k+1} in \cite{RuoyuZhangChannelTraining2025}, which can be given by
	\begin{align}
		\bm{f}(\bm\mu) & = \text{const} -  \big\| \bm{Z}_{(1)}^{\text{T}} - \big( \bm{C}^{(3)} \odot \bm{C}^{(2)} \big) \bm{C}^{(1){\text{T}}} \big\|_{F}^{2}/{\sigma^{2}_{n}}\\
		& = \text{const} - \big\| \bm{Z}_{(2)}^{\text{T}} - \big( \bm{C}^{(3)} \odot \bm{C}^{(1)} \big) \bm{C}^{(2){\text{T}}} \big\|_{F}^{2}/{\sigma^{2}_{n}}\\
		& = \text{const} -  \big\| \bm{Z}_{(3)}^{\text{T}} - \big( \bm{C}^{(2)} \odot \bm{C}^{(1)} \big) \bm{C}^{(3){\text{T}}} \big\|_{F}^{2}/{\sigma^{2}_{n}},
	\end{align}
	where $\text{const} = -N_{\text{\text{re}}} N_{\text{bs}} K \ln \big( \uppi \sigma_{n}^{2} \big)$. The resulting Fisher information matrix (FIM) is $\bm{\varOmega}(\bm\eta) \in \mathbb{C}^{5Q \times 5Q}$, and the CRB for the unknown parameters $\eta$ can be obtained via $\text{CRB}(\bm\eta) = \bm{\varOmega}^{-1}(\bm\eta)$. However, obtaining an analytical expression for the CRB is exceedingly challenging because inverting this unstructured $5Q \times 5Q$ block matrix lacks a closed-form solution, which in turn poses a significant obstacle to optimizing the CRB. 
	
	Fortunately, by employing the tensor CPD method, \eqref{AAhat}, \eqref{CChat}, and \eqref{tildeBBhat} become mutually independent. Furthermore, it should be noted that the expressions in \eqref{AAhat}, \eqref{CChat}, and \eqref{tildeBBhat} enable the derivation of the CRB because they explicitly characterize the statistical relationship between the estimated parameters and the true values under a defined perturbation model. This formulation allows the construction of a likelihood function and computation of the Fisher information matrix for all parameters. This allows the unknown parameters to be decoupled into three distinct groups  $\{ \vartheta_{q}, \varphi_{q} \}_{q=1}^{Q}$, $\{ \omega_{q}, \vartheta_{q}, \varphi_{q} \}_{q=1}^{Q}$, and $\{ \beta_{q}, \tau_{q} \}_{q=1}^{Q}$. 
	Based on \eqref{AAhat}-\eqref{tildeBBhat}, we can independently compute the CRBs of the parameters. 
	As proved in Appendix \ref{Appendix A}, the CRBs are given by
	\begin{align}
		\bm{CRB}_{\bm{\vartheta} \bm{\vartheta}} & = \big[{{\bm{F}}_{\bm{\vartheta} \bm{\vartheta}}} - {{\bm{F}}_{\bm{\vartheta} \bm{\varphi}}} \bm{F}_{\bm{\varphi} \bm{\varphi}}^{-1} \bm{F}_{\bm{\vartheta} \bm{\varphi}}^{\text{H}}\big]^{-1}
		\label{CRBUU}, \\
		\bm{CRB}_{\bm{\varphi} \bm{\varphi}}   & = \big[\bm{F}_{\bm{\varphi} \bm{\varphi}} - \bm{F}_{\bm{\vartheta} \bm{\varphi}}^{\text{H}} \bm{F}_{\bm{\vartheta} \bm{\vartheta}}^{-1} {{\bm{F}}_{\bm{\vartheta} \bm{\varphi}}}\big]^{-1}\label{CRBVV}, \\
		\bm{CRB}_{\bm{\omega} \bm{\omega}}&= \big[{{\bm{F}}_{\bm{\omega} \bm{\omega}}} - {{\bm{F}}_{\bm{\omega} {{\bm{\delta }}^{(2)}}}} {{\bm{F}}_{{{\bm{\delta }}^{(2)}}{{\bm{\delta }}^{(2)}}}} {{\bm{F}}^{\text{H}}_{\bm{\omega} {{\bm{\delta }}^{(2)}}}}\big]^{-1}\label{CRBTT},\\
		\bm{CRB}_{\bm{\tau} \bm{\tau}} &= \big[{{\bm{F}}_{\tau \tau}} - {{\bm{F}}_{\bm{\tau} \bm{\beta}}} \bm{F}_{\bm{\beta} \bm{\beta}}^{-1} {{\bm{F}}^{\text{H}}_{\bm{\tau} \bm{\beta}}}\big]^{-1}\label{CRBNN}.
	\end{align}
	
	A key strength of the CPD is its capability to decompose a high-order tensor into a sum of rank-one tensors, thus inherently isolating individual targets. Accordingly, we derive the CRB separately for each of the $Q$ targets. Without loss of generality and for notational clarity, we omit the subscript $q$ in all subsequent variables. The expressions of CRBs in \eqref{CRBUU}-\eqref{CRBNN} can be expressed as
	\begin{align}
		&CRB_{\vartheta\vartheta} = \Big[ \frac{8\uppi^2|\delta^{(1)}|^2}{\lambda_c^2\sigma_{\text{e}^{(1)}}^2} \big( \var(\mathcal{Y}) - \frac{\cov(\mathcal{X},\mathcal{Y})^2}{\var(\mathcal{X})} \big) \Big]^{-1}\label{CRBuu},
		\\
		&CRB_{\varphi\varphi} = \Big[ \frac{8\uppi^2|\delta^{(1)}|^2}{\lambda_c^2 \sigma_{\text{e}^{(1)}}^2} \big( \var(\mathcal{X}) - \frac{\cov(\mathcal{X},\mathcal{Y})^2}{\var(\mathcal{Y})} \big) \Big]^{-1}\label{CRBvv},
		\\
		&{CRB}_{\omega \omega} = \Big[\frac{8\uppi^2 T_{\text{sym}}^2 |\delta^{(2)}|^2}{\sigma_{\tilde{e}^{(2)}}^2} \var(\mathcal{T})\Big]^{-1},\label{CRBnunu}
		\\
		&{CRB}_{\tau \tau} =  \Big[ \frac{8\uppi^2 |\beta|^2  f_{\text{s}}^2| \delta^{(3)} |^2 K}{K_0^2\sigma_{\text{e}^{(3)}}^2}\var{(\mathcal{F})} \Big] ^{- 1}\label{CRBtautau}.
	\end{align}
	The parameters $\delta^{(1)}$, $\delta^{(2)}$, and $\delta^{(3)}$ are derived from $\bm{\varDelta}^{(1)}$, $\bm{\varDelta}^{(2)}$, and $\bm{\varDelta}^{(3)}$ respectively, while $\sigma_{\text{e}^{(1)}}^2$, $\sigma_{\tilde{e}^{(2)}}^2$, and $\sigma_{\text{e}^{(3)}}^2$ represent the power of error terms.\footnote{It is worth noting that the constant scaling factors in these CRB expressions are strictly independent of the structural design variables. Consequently, their exact numerical values are not required to determine the optimal resource distribution and will not introduce any additional computational complexity in practice.} The detailed derivations of \eqref{CRBuu}-\eqref{CRBtautau} are provided in Appendix \ref{Appendix B}.
	In the next section, we will present the flexible spatial-temporal-spectral optimization for $\mathcal{X},\mathcal{Y}$, $\mathcal{F}$, and $\mathcal{T}$, respectively. 
	In the following subsections, we present and prove the optimized MA positions, subcarrier distribution, and OFDM symbols distribution for wireless sensing.


	\subsection{Flexible Spatial Optimization}\label{subsectionIIIA}
	In this subsection, we focus on the flexible spatial optimization in \eqref{uu} and \eqref{vv} for positions of MAs. Since the positions of receive antennas can be flexibly adjusted and other coefficients are fixed, we transform the CRB minimization problem into optimizing the expression respect to $\mathcal{X}$ and $\mathcal{Y}$. 
	\begin{align}
		\min CRB_{\vartheta\vartheta} &\Leftrightarrow \max_{\mathcal{X},{\mathcal{Y}}} \var(\mathcal{Y}) - \frac{\cov(\mathcal{X},\mathcal{Y})^2}{\var(\mathcal{X})}\label{uu},\\
		\min CRB_{\varphi\varphi} &\Leftrightarrow \max_{\mathcal{X},{\mathcal{Y}}} \var(\mathcal{X}) - \frac{\cov(\mathcal{X},\mathcal{Y})^2}{\var(\mathcal{Y})}\label{vv}.
	\end{align}
	Since the variables in \eqref{uu} and \eqref{vv} are identical and the objective functions exhibit rotational symmetry with respect to $\mathcal{X}$ and $\mathcal{Y}$, we jointly optimize \eqref{uu} and \eqref{vv}.
	
	We consider MAs located within a fixed region $\mathcal{C}$. To prevent severe mutual coupling between antennas, we enforce a minimum separation distance of $d$ between any two antennas. Thus, \eqref{uu} and \eqref{vv} can be expressed as
\begin{subequations}\label{maxuuvv}
		\begin{align}
		\max_{\mathcal{X},{\mathcal{Y}}} \enspace & \var(\mathcal{Y}) - \frac{\cov(\mathcal{X},\mathcal{Y})^2}{\var(\mathcal{X})}+ \var(\mathcal{X}) - \frac{\cov(\mathcal{X},\mathcal{Y})^2}{\var(\mathcal{Y})} \\
		\text{s.t.} \enspace
			&(\mathcal{X},\mathcal{Y}) \in \mathcal{C}, \\
			&\sqrt{(x_{n_1}-x_{n_2})^2 + (y_{n_1}-y_{n_2})^2} \geq d. \label{47c}
	\end{align}
\end{subequations}
		Solving \eqref{maxuuvv} presents significant challenges. On one hand, \eqref{47c} is non convex potentially leading to local optima. On the other hand irregular boundaries of the movable region may preclude analytical solutions. Fortunately,\cite{maMovableAntennaEnhanced2024} has derived analytical solutions for optimal antenna arrays in circular regions and developed computational methods for general regions. Specifically this highly non convex problem is efficiently solved by alternating the optimization variables between the $x$ and $y$ directions. In each alternating step the successive convex approximation method is utilized to construct a convex surrogate function for the non convex minimum distance constraint iteratively guiding the antenna coordinates to a stationary point. We will compute the optimal MA positions following this alternating approach in \cite{maMovableAntennaEnhanced2024}.

\subsection{Flexible Temporal Optimization}\label{subsectionIIIB}
In this subsection, we focus on the flexible temporal optimization for the OFDM symbol distribution. Similar to the process in Section \ref{subsectionIIIA}, we transform the CRB minimization into optimizing the expression with respect to $\mathcal{T}$.
\begin{align}
	\min \enspace CRB_{\omega\omega} \enspace\Leftrightarrow\enspace \max_{\mathcal{T}} \enspace\var(\mathcal{T}) \label{nunu}.
\end{align}
Different from the elements in $\mathcal{X}$ and $\mathcal{Y}$, the set $\mathcal{T}$ must be a set of positive integers. Then the optimization problem in \eqref{nunu} can be formulated as
\begin{subequations}\label{maxM}
	\begin{align}
	 \max_{\mathcal{T}} \enspace &\var(\mathcal{T}) \\
	\text{s.t.}  \enspace & \mathcal{T}(m_1) \ne \mathcal{T}(m_2), \quad m_1, m_2 = 1, \ldots, M \\
		& 1 \le \mathcal{T}(m) \le M_0, \quad m = 1, \ldots, M \\
		& \mathcal{T}(m) \in \mathbb{Z}^{+},\quad m = 1, \ldots, M.
\end{align}
\end{subequations}
Although the objective function of the aforementioned optimization problem is convex, the constraints that all variables must be distinct from each other and take integer values make this problem particularly challenging to solve. Fortunately, Theorem \ref{Thoerem1} provides a closed-form solution to address a class of optimization problems sharing a similar form. Let $\mathcal{T}^{\star}$ denote the optimal solution to \eqref{maxM}.

\begin{theorem}\label{Thoerem1}
	The optimal OFDM symbol distribution that minimizes the root mean square error of Doppler estimation is given by
	\begin{align}
		\mathcal{T}(m)^{\star} = \begin{cases}
			1,2,\ldots,\bar{M} & m \leq \bar{M}, \\
			M_0 - \bar{M} + 1,\ldots,M_0 & m > \bar{M}.
		\end{cases}
	\end{align}
	where $\bar{M} = \lfloor M/2 \rfloor$.
\end{theorem}
\begin{proof}
	See Appendix \ref{Appendix C}.
\end{proof}

Theorem \ref{Thoerem1} demonstrates that to minimize the CRB for Doppler estimation, the training OFDM symbols should be divided into two groups. Specifically, half of the training OFDM symbols should occupy the lower end of the total OFDM symbol range, while the other half should occupy the upper end. It is straightforward to understand that the distribution proposed in Theorem \ref{Thoerem1} maximizes the variance by spreading the OFDM symbols as widely as possible.

Furthermore, given the number of OFDM symbols $M$, increasing the total number of OFDM symbols $M_0$ can effectively reduce $CRB_{\omega\omega}$.

\subsection{Flexible Spectral Optimization}\label{subsectionIIIC}

In this section, we focus on the flexible spectral optimization of the subcarrier distribution. Similar to the reasoning in Section \ref{subsectionIIIB}, we transform the CRB minimization into optimizing the expression with respect to $\mathcal{F}$.
\begin{align}
	\min \enspace CRB_{\tau\tau} \enspace\Leftrightarrow\enspace \max_{\mathcal{F}} \enspace\var(\mathcal{F}) \label{tautau}.
\end{align}
Following the rationale established in Section \ref{subsectionIIIB}, the optimization problem in \eqref{tautau} can be formulated as
\begin{subequations}\label{maxK}
	\begin{align}
	 \max_{\mathcal{F}} \enspace &\var(\mathcal{F})\\
	\text{s.t.} \enspace
		& \mathcal{F}(k_1) \ne \mathcal{F}(k_2), \quad k_1, k_2 = 1, \ldots, K, \\
		& 1 \le \mathcal{F}(k) \le K_0, \quad k = 1, \ldots, K, \\
		& \mathcal{F}(k) \in \mathbb{Z}^{+},\quad k = 1, \ldots, K.
\end{align}
\end{subequations}
Let $\mathcal{F}^{\star}$ denote the optimal solution to \eqref{maxK}. It is noteworthy that \eqref{maxK} is similar to the form of \eqref{maxM}. Consequently, we can likewise employ Theorem \ref{Thoerem1} to obtain the optimal solution $\mathcal{F}^{\star}$.

	To this end, we have shown the proposed flexible spatial-temporal-spectral optimization-based tensor algorithm, which is summarized in Algorithm \ref{Algorithm 1}. Notably, the proposed estimation algorithm imposes no additional hardware requirements {compared} to \cite{RuoyuZhangISACwithMIMO2024,RuoyuZhangChannelTraining2025}, nor does it increase computational complexity. Furthermore, the tensor decomposition approach enables independent estimation of each target parameter, thereby overcoming the performance degradation caused by target coupling in conventional algorithms under conditions of flexible spatial-temporal-spectral optimization. 
	\begin{algorithm}[t]
		\renewcommand{\algorithmicrequire}{\textbf{Input }}  
		\renewcommand{\algorithmicensure}{\textbf{Output }}  
		\caption{Proposed Flexible Spatial-Temporal-Spectral Optimization-Based Tensor Algorithm}
		\label{Algorithm 1}
		\begin{algorithmic}[1]
			\Require 
			The received signal tensor $\bm{\mathcal{Z}}$, precoding matrix $\bm{P}$, the number of target sources $Q$, the movable rectangular region $\mathcal{C}$, the number of MAs $N_{\text{\text{re}}}$, the total number of OFDM carriers $K_0$ and OFDM symbols $M_0$, the number of symbols $K$ and subcarriers $M$.
			
			\Statex {\textbf{Stage 1 Flexible spatial-temporal-spectral optimization}}
			\State Optimize $\mathcal{X},\mathcal{Y},\mathcal{F},\mathcal{T}$ via \eqref{maxuuvv}, \eqref{maxK}, and \eqref{maxM};
			\Statex {\textbf{Stage 2 Parameter estimation}}
			\State Compute the factor matrices $\bm{C}^{(1)},\bm{C}^{(2)},\bm{C}^{(3)}$ via \eqref{AAhat}-\eqref{CChat};
			\For {each $q\in [1,\ldots,Q]$ }
			\State Estimate $\{\hat\theta_q,\hat\phi_q,\hat\omega_q,\hat\tau_q\}$ via \eqref{hatthetaphi}, \eqref{hatDop}, and \eqref{hattau};
			\EndFor
			\State Estimate $\hat{\bm\beta}$ via \eqref{hatbeta};
			\Ensure Return estimated parameters $\{\hat\theta_q,\hat\phi_q,\hat\omega_q,\hat\tau_q,\hat\beta_q\}.$
		\end{algorithmic}
	\end{algorithm}

	\section{Simulation Results}\label{section VI}

		\begin{figure}[t]
			\centering
			\includegraphics[width=0.75\linewidth]{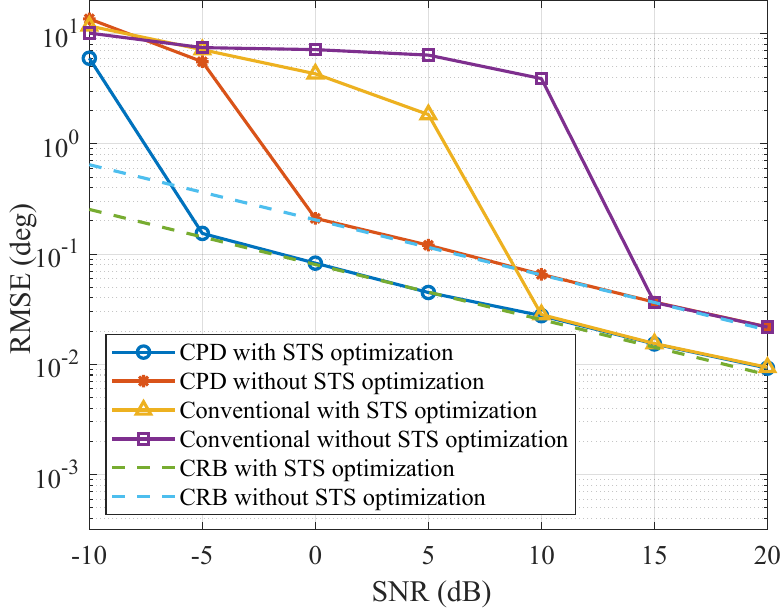}
			\caption{RMSE of elevation angles versus SNR.}
			\label{theta_SNR}
		\end{figure}
		\begin{figure}[t]
			\centering
			\includegraphics[width=0.75\linewidth]{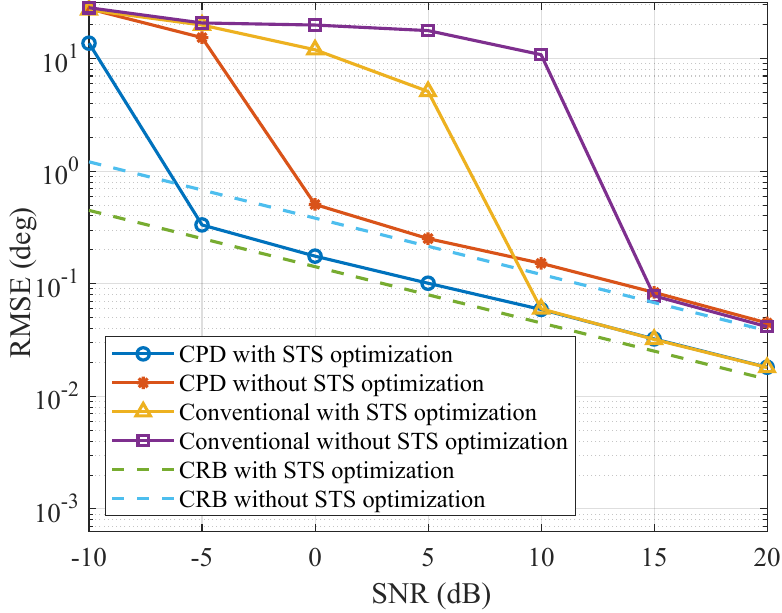}
			\caption{RMSE of azimuth angles versus SNR.}
			\label{phi_SNR}
		\end{figure}
		\begin{figure}[t]
			\centering
			\includegraphics[width=0.75\linewidth]{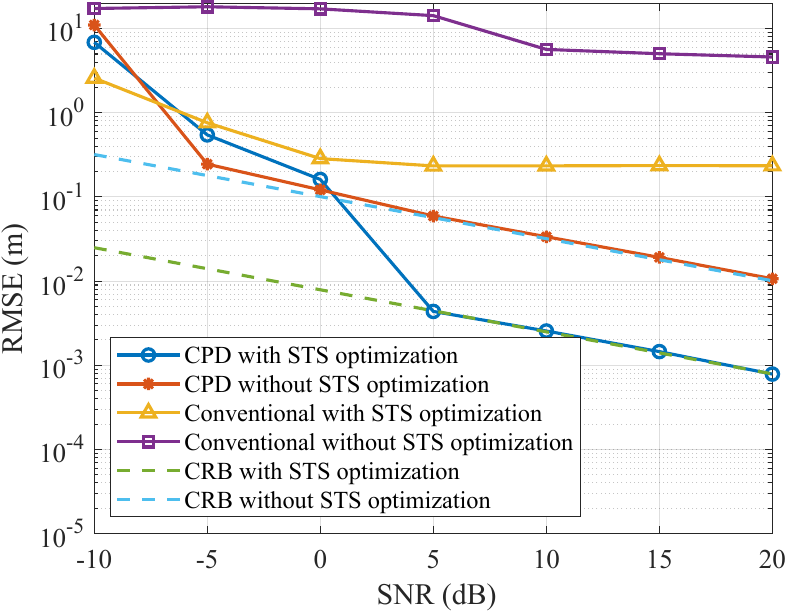}
			\caption{RMSE of target ranges versus SNR.}
			\label{tau_SNR}
		\end{figure}
		\begin{figure}[t]
			\centering
			\includegraphics[width=0.75\linewidth]{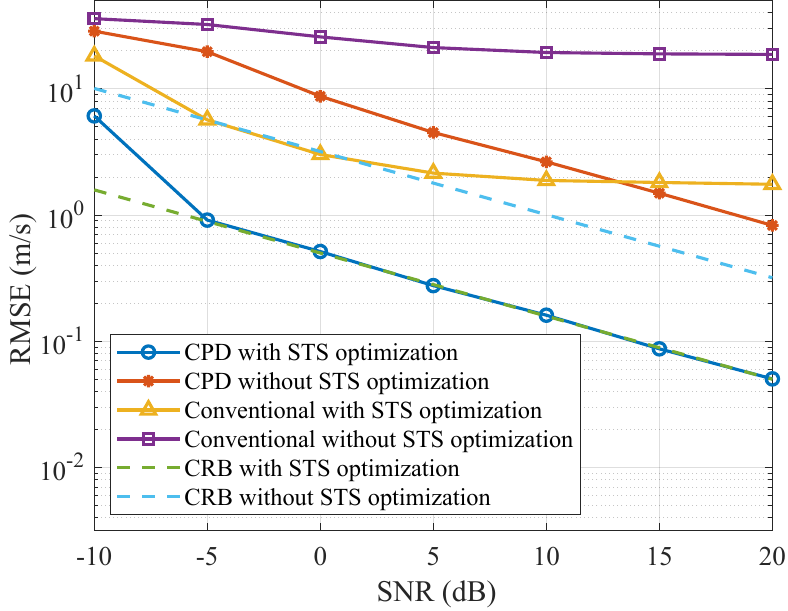}
			\caption{RMSE of velocities versus SNR.}
			\label{Doppler_SNR}
		\end{figure}
		\begin{figure}[t]
			\centering
			\includegraphics[width=0.75\linewidth]{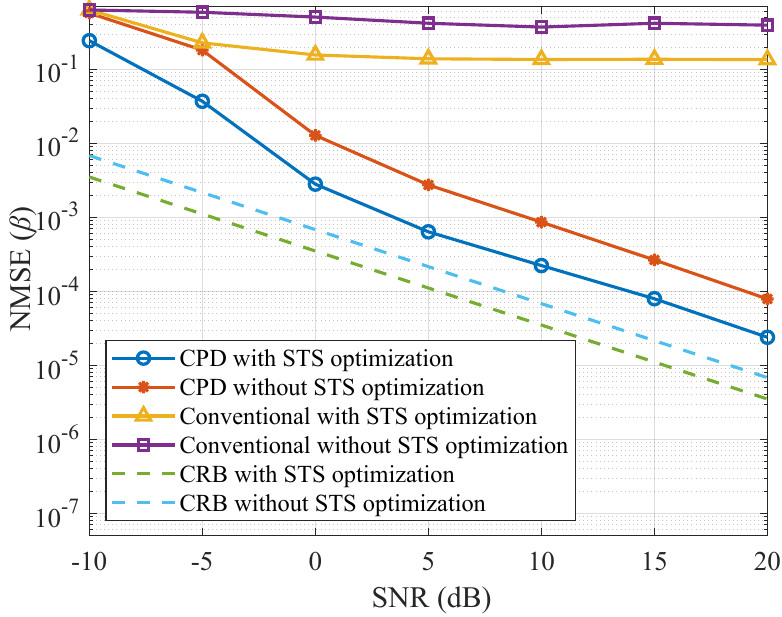}
			\caption{NMSE of reflection coefficients versus SNR.}
			\label{falpha_SNR}
		\end{figure}
		\begin{figure}[t]
			\centering
			\includegraphics[width=0.77\linewidth]{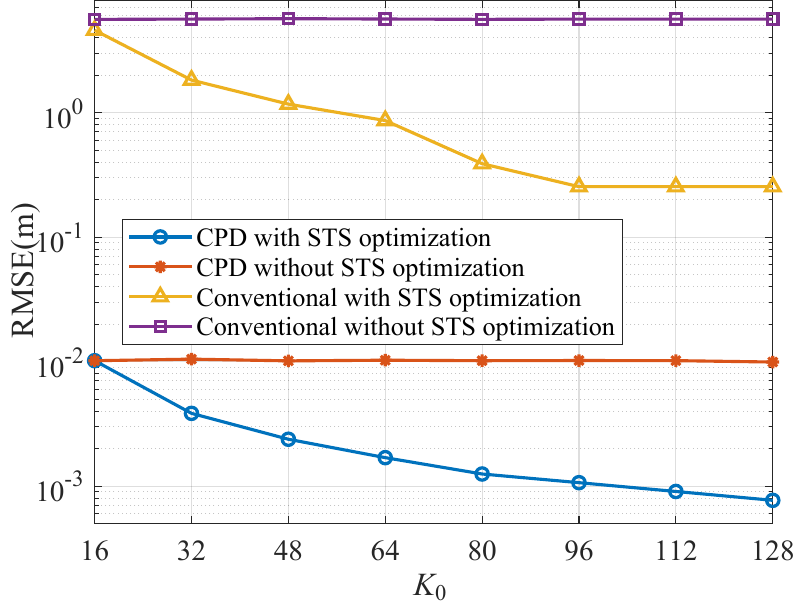}
			\caption{RMSE of target distances versus the total number of subcarriers.}
			\label{tau_T_0}
		\end{figure}
		\begin{figure}[t]
			\centering
			\includegraphics[width=0.77\linewidth]{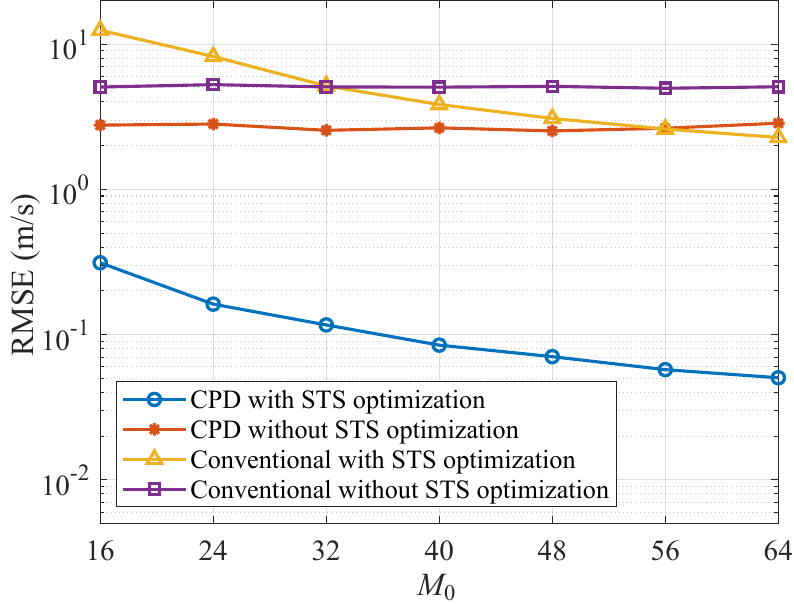}
			\caption{RMSE of velocities versus the total number of OFDM symbols.}
			\label{Doppler_T_0}
		\end{figure}

		In this section, we evaluate the performance of the proposed tensor-based flexible spatial-temporal-spectral optimization through numerical simulations. We consider a BS equipped with $N_{\text{bs}}=36$ {fixed-position transmit} antennas with $N_{\text{bs,}x}=6$, $N_{\text{bs,}y}=6$, {and} $N_{\text{\text{re}}}=36$ receive {MAs}. 
		The carrier frequency is $f_c=28\text{GHz}$ and the bandwidth is $f_{\text{s}}=100\text{MHz}$ with $K_0=128$ and $K=16$. The CP duration in one OFDM symbol is $T_{\text{cp}}={1}/{(2\Delta f)}$ with $M_0=64$ and $M=16$. The BS is detecting $Q=3$ targets, where 
		the AoAs (deg) and AoDs (deg) follow $\mathcal{U}(0,2\uppi)$,
		the distances $R$ (m) follows $\mathcal{U}(0,48)$,
		and the velocities $V$ (m/s) follows $\mathcal{U}(-30,30)$
		where the minus sign indicates the direction opposite to the positive direction. The signal-to-noise ratio (SNR) is given by 
		$\text{SNR}=||\bm{Z}_k-\bm{N}_k||_F^2/||\bm{N}||_F^2$.
		The RMSE can be defined as 
		$\text{RMSE}=\sqrt{{\sum_{i=1}^{I}\sum_{q=1}^{Q}(\hat{{\eta}}_{iq}-\eta_{q})^2}/(IQ)}$,
{		where $I$ denotes the number of Monte Carlo trials. $\eta_q\in \{\vartheta_q,\varphi_q,R_q,V_q\}_{q=1}^{Q}$ denotes the true parameter and $\hat{{\eta}}_{iq}$ denotes the $i$-th trial $q$-th estimated parameter, respectively.}
		The estimation of the reflection coefficient $\bm{\beta}$ performance at the $m$-th OFDM are calculated by the {normalized mean square error (NMSE)}, which can be given by
		$\text{NMSE}(\bm{\beta})={\|\bm{\beta}-\hat{\bm{\beta}}\|_F^2}/{\|\bm{\beta}\|_F^2}$.
		The proposed method and benchmarks are as follows:
		\begin{itemize}
			\item \textbf{CPD with STS optimization}: The proposed tensor-based method with the proposed STS optimization.
			\item \textbf{CPD without STS optimization}: The proposed tensor-based method without STS optimization.
			\item \textbf{Conventional with STS optimization}: The conventional MUSIC and MF methods \cite{Liu2020Joint} with STS optimization.\footnote{The conventional uniform steering vectors corresponding to the antenna positions, subcarrier and OFDM symbols distribution are placed by the proposed STS optimized steering vectors.}
			\item \textbf{Conventional without STS optimization}: The conventional MUSIC and MF methods without STS optimization.
			\item \textbf{CRB with STS optimization}: The theoretical CRB of the parameter estimation problem with STS optimization.
			\item \textbf{CRB without STS optimization}: The theoretical CRB of the parameter estimation problem without STS optimization.
		\end{itemize}
		
	\Cref{theta_SNR,phi_SNR} illustrate the RMSE of elevation and azimuth angle estimation versus SNR for various algorithms. By applying the proposed STS optimization, the theoretical CRBs for estimating both angles are significantly lowered. For the actual estimation performance, the conventional algorithms completely fail at low SNR regimes because severe mutual interference among multiple targets drastically degrades the accuracy. In stark contrast, the proposed algorithm based on tensor decomposition achieves substantial performance gains over conventional methods. It successfully approaches the optimized CRBs even at low SNR levels. This remarkable enhancement stems from the inherent capability of tensor decomposition to isolate individual targets and completely eliminate mutual interference among targets.
		
		\Cref{tau_SNR} presents the RMSE performance of the range estimation versus SNR. The STS optimization effectively reduces the CRB for target range estimation. The results indicate that the conventional algorithm without STS optimization performs poorly and struggles to acquire accurate delay estimation. Although incorporating STS optimization into the conventional algorithm provides a noticeable performance lift, it still suffers from an error floor at high SNR regions due to the unresolved mutual interference. Conversely, the proposed algorithm based on tensor decomposition overcomes this saturation bottleneck. By further integrating the STS optimization, the proposed method fully unlocks the spatial, temporal, and spectral degrees of freedom, achieving the lowest RMSE and tightly approaching the optimized CRB.
		
\Cref{Doppler_SNR} displays the RMSE performance of velocity estimation versus SNR for different algorithms. With the implementation of STS optimization, the corresponding CRB exhibits a significant decrease. The results explicitly reveal that conventional algorithms fail to estimate velocities accurately regardless of whether STS optimization is applied, which is primarily caused by the severe mutual interference. In contrast, the proposed algorithms based on tensor decomposition maintain robust estimation capabilities. Furthermore, the combination of the proposed algorithm and STS optimization yields the most superior results, demonstrating a distinct performance gap compared to all baseline methods and tightly approaching the lowered theoretical bound.
		
		\Cref{falpha_SNR} shows the NMSE performance of reflection coefficient estimation versus SNR. Because the reflection coefficient estimation heavily relies on previously estimated parameters including the AoA, AoD, range, and velocity of each target, any estimation error from previous stages propagates and degrades the final NMSE. Consequently, the conventional algorithms cannot accurately estimate the reflection coefficient, presenting a flat performance curve even at high SNR levels. Benefiting from the highly accurate parameter estimation provided by the aforementioned stages, the proposed algorithm based on tensor decomposition successfully circumvents this error propagation. Empowered by the STS optimization, the proposed algorithm achieves the most outstanding NMSE performance and strictly follows the minimized CRB.
		
		\Cref{tau_T_0} presents the RMSE of range estimation versus the total number of subcarriers for the algorithms based on tensor decomposition, given a fixed subcarrier spacing and an SNR of $20$ dB. The algorithm without the STS optimization exhibits a constant performance curve regardless of the total number of subcarriers. This occurs because the unoptimized contiguous subcarrier allocation fails to fully exploit the newly added frequency resources. In contrast, the proposed algorithm equipped with the STS optimization demonstrates progressively improved performance as the total number of subcarriers increases. By widely spreading the selected subcarriers across the entire available band, the STS optimization effectively maximizes the frequency aperture. Consequently, the proposed algorithm with STS optimization maintains significantly better performance compared with the unoptimized baselines.
		
\Cref{Doppler_T_0} shows the RMSE of velocity estimation versus the total number of OFDM symbols for the algorithms based on tensor decomposition with an SNR of $20$ dB. The algorithm without STS optimization maintains a consistent error level across different total numbers of OFDM symbols due to its restrictive contiguous temporal allocation. However, the proposed algorithm enhanced by the STS optimization exhibits progressively improved performance as the total number of OFDM symbols increases. By distributing the selected OFDM symbols toward the two ends of the available time frame, the STS optimization maximizes the effective observation time and refines the Doppler resolution. Similar to the range estimation scenario, the proposed algorithm with STS optimization consistently achieves significantly better performance than the conventional algorithms.

		\section{Conclusion}\label{section VII}
		
	This paper has proposed a novel wireless sensing framework for MIMO-OFDM ISAC systems to {achieve} flexible optimization across the spatial-temporal-spectral dimensions. We have established a monostatic sensing model where antenna positions, OFDM symbol allocation, and subcarrier assignment can be dynamically configured, and developed a tensor {decomposition-based} approach that transforms target parameter estimation into a parallelizable process. Our derivation of the CRB reveals fundamental relationships between estimation accuracy and system configuration, demonstrating that the performance bounds for azimuth/elevation angles, velocities, and ranges are intrinsically determined by array geometry and time-frequency resource distribution. Building on this theoretical foundation, we have obtained a globally optimal solution for joint spatial-temporal-spectral resource allocation that simultaneously minimizes both the CRB and mean square error of estimation. Extensive numerical results validate that the proposed framework significantly outperforms conventional {sensing} systems with fixed-position antennas and static resource allocation, {showcasing its strong potential} to enhance sensing accuracy in next-generation {wireless} networks.
		
		\appendices
		\section{Derivation of CRB}\label{Appendix A}
		In this {appendix}, we present the {derivation} for the submatrices in equations \eqref{CRBUU}-\eqref{CRBNN}. 
		\subsection{{Derivation of \textbf{CRB}$_{\vartheta\vartheta}$ and \textbf{{CRB}}$_{\varphi\varphi}$ }}
		For \eqref{AAhat}, by vectorizing $\bm{\hat {C}}^{(1)}$, we obtain
		\begin{align}
			\text{vec} (\bm{\hat {C}}^{(1)})=
			{{\bm{\hat{c}}}^{(1)}} = \big[ {{\bm{I}}_{Q \times Q}} \otimes {{\bm{C}}^{(1)}} \big]{{\bm{\delta }}^{(1)}} + {{\bm{e}}^{(1)}} \label{C1vec},
		\end{align}
		where ${{\bm{\delta }}^{(1)}} = \operatorname{vec}\big( {{\bm{\varDelta }}^{(1)}} \big)$, ${{\bm{e}}^{(1)}} = \operatorname{vec}\big( {{\bm{E}}^{(1)}} \big) $. From \eqref{C1vec}, ${{\bm{\hat{c}}}^{(1)}}$ can be viewed as a distribution with the mean vector ${{\bm{\mu }}^{(1)}}$ and covariance matrix ${{\bm{H}}^{(1)}}$, which are given by
		\begin{align}
			{{\bm{\mu }}^{(1)}} & = \big[ {{\bm{I}}_{Q \times Q}} \otimes {{\bm{C}}^{(1)}} \big]{{\bm{\delta }}^{(1)}}, \\
			{{\bm{H}}^{(1)}} & = \mathbb{E}\big\{ {{\bm{e}}^{(1)}}{{\bm{e}}^{(1)\text{H}}} \big\} = \sigma_{e^{(1)}}^{2}{{\bm{I}}_{N_{\text{re}}Q \times N_{\text{re}}Q}}.
		\end{align}
		
		Since we are interested in the parameters, we can construct a Fisher Information Matrix (FIM) based on \eqref{C1vec}
		\begin{align}
			{{\bm{F}}^{(1)}}(\bm{\vartheta}, \bm{\varphi}) = \big[ \begin{matrix}
				{{\bm{F}}_{\bm{\vartheta} \bm{\vartheta}}} & {{\bm{F}}_{\bm{\vartheta} \bm{\varphi}}} \\
				\bm{F}_{\bm{\vartheta} \bm{\varphi}}^{\text{H}} &  \bm{F}_{\bm{\varphi} \bm{\varphi}} \\
			\end{matrix} \big] \in \mathbb{C}^{2Q \times 2Q},
		\end{align}
		where the $(q,q')$-th element of ${{\bm{F}}^{(1)}}(\bm{\vartheta}, \bm{\varphi})$ can be given by
		\begin{align}
			{{\big[ {{\bm{F}}^{(1)}} \big]}_{q,q'}} &= 2\mathfrak{R} \big[ \frac{\partial {{\bm{\mu }}^{(1)\text{H}}}}{\partial \eta_{q}} [{{\bm{H}}^{(1)}}]^{-1} \frac{\partial {{\bm{\mu }}^{(1)}}}{\partial \eta_{q'}} \big]\label{yuan}.
		\end{align}
		
		The mean vector ${{\bm{\mu }}^{(1)}}$ is a function of $\bm{\vartheta}$ and $\bm{\varphi}$. Its first-order derivatives with respect to $\bm{\vartheta}$ and $\bm{\varphi}$ are given by
		\begin{align}
			\frac{\partial {{\bm{\mu }}^{(1)}}}{\partial \bm{\vartheta}} & = {{\big[ \bm{\varLambda}{{_{1}^{(1)\text{H}}}} \bm{C}{{_{\bm{\vartheta}}^{(1)\text{H}}}}, \ldots, \bm{\varLambda}{{_{Q}^{(1)\text{H}}}} \bm{C}{{_{\bm{\vartheta}}^{(1)\text{H}}}} \big]}^{\text{T}}}, \\
			\frac{\partial {{\bm{\mu }}^{(1)}}}{\partial \bm{\varphi}}  & = {{\big[ \bm{\varLambda}{{_{1}^{(1){\textsc{H}}}}} \bm{C}{{_{\bm{\varphi}}^{(1)\text{H}}}}, \ldots, \bm{\varLambda}{{_{Q}^{(1H}}} \bm{C}{{_{\bm{\varphi}}^{(1)\text{H}}}} \big]}^{\text{T}}},
		\end{align}
		where $\bm{\varLambda}_{q}^{(1)} = \diag ( {{[{{\bm{\varDelta }}^{(1)}}]}_{(:,q)}} )$. The matrices $\bm{C}_{\bm{\vartheta}}^{(1)}$ and $\bm{C}_{\bm{\varphi}}^{(1)}$ consist of the derivatives of the columns of $\bm{C}^{(1)}$ with respect to the corresponding $\bm{\vartheta}$ and $\bm{\varphi}$, respectively, and are given by
		\begin{align}
			\bm{C}_{\bm{\vartheta}}^{(1)} & = \big[ {{\frac{\partial {{\bm{a}}_{\text{\text{re}}}}}{\partial \vartheta} \big|}_{\vartheta = \vartheta_{1}}}, \ldots, {{\big. \frac{\partial {{\bm{a}}_{\text{\text{re}}}}}{\partial \vartheta} \big|}_{\vartheta = \vartheta_{Q}}} \big] \in \mathbb{C}^{N_{\text{re}} \times Q} \label{11dot},\\
			\bm{C}_{\bm{\varphi}}^{(1)}  & = \big[ {{\big. \frac{\partial {{\bm{a}}_{\text{\text{re}}}}}{\partial \varphi} \big|}_{\varphi = \varphi_{1}}}, \ldots, {{\big. \frac{\partial {{\bm{a}}_{\text{\text{re}}}}}{\partial \varphi} \big|}_{\varphi = \varphi_{Q}}} \big] \in \mathbb{C}^{N_{\text{re}} \times Q}\label{22dot}.
		\end{align}
		
		Using \eqref{yuan}, we obtain the following expressions:
		\begin{align}
			&\!\!{{\bm{F}}_{\bm{\vartheta} \bm{\vartheta}}} =\frac{2}{\sigma_{e^{(1)}}^{2}} \mathfrak{R} \big[ \sum_{q=1}^{Q} \bm{\varLambda}{{_{q}^{(1)\text{H}}}} \bm{C}{{_{\bm{\vartheta}}^{(1)\text{H}}}} \bm{C}_{\bm{\vartheta}}^{(1)} \bm{\varLambda}_{q}^{(1)} \big] \label{F11},\\
			&\!\!{{\bm{F}}_{\bm{\varphi} \bm{\varphi}}} = \frac{2}{\sigma_{e^{(1)}}^{2}} \mathfrak{R} \big[ \sum_{q=1}^{Q} \bm{\varLambda}{{_{q}^{(1)\text{H}}}} \bm{C}{{_{\bm{\varphi}}^{(1)\text{H}}}} \bm{C}_{\bm{\varphi}}^{(1)} \bm{\varLambda}_{q}^{(1)} \big]\label{F22}, \\
			\!\!&{{\bm{F}}_{\bm{\vartheta} \bm{\varphi}}} = \frac{2}{\sigma_{e^{(1)}}^{2}}\mathfrak{R} \big[ \sum_{q=1}^{Q} \bm{\varLambda}{{_{q}^{(1)\text{H}}}} \bm{C}{{_{\bm{\vartheta}}^{(1)\text{H}}}} \bm{C}_{\bm{\varphi}}^{(1)} \bm{\varLambda}_{q}^{(1)} \big]\label{F12}.
		\end{align}
		
		Applying the block matrix inversion formula, we obtain the CRBs
		\begin{align}
			\bm{CRB}_{{\bm{\vartheta}} {\bm{\vartheta}}} & = \big[{{\bm{F}}_{\bm{\vartheta} \bm{\vartheta}}} - {{\bm{F}}_{\bm{\vartheta} \bm{\varphi}}} \bm{F}_{\varphi \varphi}^{-1} \bm{F}_{\bm{\vartheta} \bm{\varphi}}^{\text{H}}\big]^{-1}, \\
			\bm{CRB}_{\bm{\varphi} \bm{\varphi}}   & = \big[\bm{F}_{\bm{\varphi} \bm{\varphi}} - \bm{F}_{\bm{\vartheta} \bm{\varphi}}^{\text{H}} \bm{F}_{\bm{\vartheta} \bm{\vartheta}}^{-1} {{\bm{F}}_{\bm{\vartheta} \bm{\varphi}}}\big]^{-1}.
		\end{align}
		
		\subsection{\textcolor{black}{Derivation of  $\textbf{{CRB}}_{{\omega}\omega}$} }
		For \eqref{tildeBBhat}, by vectorizing ${{\bm{\tilde{\hat{C}}}^{(2)}}}$, we obtain
		\begin{align}
			\text{vec}({{\bm{\tilde{\hat{C}}}}^{(2)}})=
			{{\bm{\tilde{\hat{c}}}}^{(2)}} = \big[ {{\bm{I}}_{Q \times Q}} \otimes {{\bm{\tilde{C}}}^{(2)}} \big]{{\bm{\delta }}^{(2)}} + {{\bm{\tilde{e}}}^{(2)}} \label{33cc},
		\end{align}
		where ${{\bm{\delta }}^{(2)}} = \operatorname{vec}\big( {{\bm{\varDelta }}^{(2)}} \big)$, ${{\bm{\tilde{e}}}^{(2)}} = \operatorname{vec}\big( {{\bm{\tilde{E}}}^{(2)}} \big) $, and ${{\bm{\tilde{\hat{c}}}}^{(2)}}$ can be viewed as a distribution with mean vector ${{\bm{\mu }}^{(2)}} = [ {{\bm{I}}_{Q \times Q}} \otimes {{\bm{\tilde{C}}}^{(2)}} ]{{\bm{\delta }}^{(2)}}$ and covariance matrix ${{\bm{H}}^{(2)}} = \mathbb{E}\{ {{\bm{\tilde{e}}}^{(2)}}{{\bm{\tilde{e}}}^{(2)\text{H}}} \} = \sigma_{\tilde{e}^{(2)}}^{2}{{\bm{I}}_{MQ \times MQ}}$.
		Accordingly, the FIM with respect to $\omega$ and ${{\bm{\delta }}^{(2)}}$ are given by
{		\begin{align}
			\!{{\bm{F}}^{(2)}}(\bm{\omega}, {{\bm{\delta }}^{(2)}}) = \big[ \begin{matrix}
				{{\bm{F}}_{\bm{\omega} \bm{\omega}}} \!\!&\!\! {{\bm{F}}_{\bm{\omega} {{\bm{\delta }}^{(2)}}}} \\
				{{\bm{F}}^{\text{H}}_{\bm{\omega} {{\bm{\delta }}^{(2)}}}} \!\!&\!\! {{\bm{F}}_{{{\bm{\delta }}^{(2)}}{{\bm{\delta }}^{(2)}}}} \\
			\end{matrix} \big] \in \mathbb{C}^{(2Q^2+Q) \times (2Q^2+Q)},
		\end{align}}
		where
		\begin{align}
			\!\!{{\bm{F}}_{{{\bm{\delta }}^{(2)}}{{\bm{\delta }}^{(2)}}}} & = \big[ \begin{matrix}
				{{\bm{F}}_{\mathfrak{R}[{{\bm{\delta }}^{(2)}}]\mathfrak{R}[{{\bm{\delta }}^{(2)}}]}} \!\!&\!\! {{\bm{F}}_{\mathfrak{R}[{{\bm{\delta }}^{(2)}}]\mathfrak{I}[{{\bm{\delta }}^{(2)}}]}} \\
				{{\bm{F}}^{\text{H}}_{\mathfrak{R}[{{\bm{\delta }}^{(2)}}]\mathfrak{I}[{{\bm{\delta }}^{(2)}}]}} \!\!&\!\! {{\bm{F}}_{\mathfrak{I}[{{\bm{\delta }}^{(2)}}]\mathfrak{I}[{{\bm{\delta }}^{(2)}}]}} \\
			\end{matrix} \big] \in \mathbb{C}^{2Q^2 \times 2Q^2}, \\
			{{\bm{F}}_{\bm{\omega} {{\bm{\delta }}^{(2)}}}} & = \big[ {{\bm{F}}_{\bm{\omega} \mathfrak{R}[{{\bm{\delta }}^{(2)}}]}}, {{\bm{F}}_{\bm{\omega} \mathfrak{I}[{{\bm{\delta }}^{(2)}}]}} \big].
 		\end{align}
		The first-order derivative of ${{\bm{\mu }}^{(2)}}$ with respect to $\bm{\omega}$ is given by
		\begin{align}
			\frac{\partial {{\bm{\mu }}^{(2)}}}{\partial \bm{\omega}} = {{\big[ \bm{\varLambda}_{1}^{(2)\text{H}} \bm{\tilde{C}}_{\bm{\omega}}^{(2)\text{H}}, \ldots, \bm{\varLambda}_{Q}^{(2)\text{H}} \bm{\tilde{C}}_{\bm{\omega}}^{(2)\text{H}} \big]}^{\text{T}}},
		\end{align}
		where $\bm{\varLambda}_{q}^{(2)} = \operatorname{diag}( {{[{{\bm{\varDelta }}^{(2)}}]}_{(:,q)}} )$. The matrix $\bm{\tilde{C}}_{\bm{\omega}}^{(2)}$ consists of the derivatives of the columns of $\bm{\tilde{C}}^{(2)}$ with respect to the corresponding $\bm{\omega}$, and is given by
		\begin{align}
			\bm{\tilde{C}}_{\bm{\omega}}^{(2)} = \big[ {{\big. \frac{\partial {{\bm{a}}_{\text{do}}}}{\partial \omega} \big|}_{\omega = \omega_{1}}}, \ldots, {{\big. \frac{\partial {{\bm{a}}_{\text{do}}}}{\partial \omega} \big|}_{\omega = \omega_{Q}}} \big] \in \mathbb{C}^{M \times Q}.
		\end{align}
		Similar to \eqref{yuan}, we obtain the following expressions
		\begin{align}
			\!\!{{\bm{F}}_{\!\bm{\omega}  \bm{\omega}}} = \frac{2}{\sigma_{\tilde{e}^{(2)}}^{2}} \mathfrak{R} \big[ \sum_{q=1}^{K N_{\text{re}}} \bm{\varLambda}_{q}^{(2)\text{H}} \bm{\tilde{C}}_{\bm{\omega}}^{(2)\text{H}} \bm{\tilde{C}}_{\bm{\omega}}^{(2)} \bm{\varLambda}_{q}^{(2)} \big].
		\end{align}
		Since ${{\bm{\delta }}^{(2)}}$ is a complex-valued vector, we consider derivatives with respect to its real and imaginary parts separately
		\begin{align}
			\frac{\partial {{\bm{\mu }}^{(2)}}}{\partial \mathfrak{R}[{{\bm{\delta }}^{(2)}}]}  = {{\bm{I}}_{Q \times Q}} \otimes {{\bm{\tilde{C}}}^{(2)}} ,\quad\!\!\!\!
			\frac{\partial {{\bm{\mu }}^{(2)}}}{\partial \mathfrak{I}[{{\bm{\delta }}^{(2)}}]}  = \text{j}{{\bm{I}}_{Q \times Q}} \otimes {{\bm{\tilde{C}}}^{(2)}}.
		\end{align}
		
		Thus, we obtain
		\begin{align}
			&\!\!\!\!\!{{\bm{F}}_{\mathfrak{R}[{{\bm{\delta }}^{(2)}}]\mathfrak{R}[{{\bm{\delta }}^{(2)}}]}} = \frac{2}{\sigma_{\tilde{e}^{(2)}}} \mathfrak{R} \big[ {{\bm{I}}_{Q \times Q}} \otimes \big( {{\bm{\tilde{C}}}^{(2)}}{{\bm{\tilde{C}}}^{(2)\text{H}}} \big) \big], \\
			&\!\!\!\!\!{{\bm{F}}_{\mathfrak{I}[{{\bm{\delta }}^{(2)}}]\mathfrak{I}[{{\bm{\delta }}^{(2)}}]}} = -\frac{2}{\sigma_{\tilde{e}^{(2)}}} \mathfrak{R} \big[ {{\bm{I}}_{Q \times Q}} \otimes \big( {{\bm{\tilde{C}}}^{(2)}}{{\bm{\tilde{C}}}^{(2)\text{H}}} \big) \big], \\
			&\!\!\!\!\!{{\bm{F}}_{\mathfrak{R}[{{\bm{\delta }}^{(2)}}]\mathfrak{I}[{{\bm{\delta }}^{(2)}}]}} = \frac{2}{\sigma_{\tilde{e}^{(2)}}} \mathfrak{R} \big[ {{\bm{I}}_{Q \times Q}} \otimes \big( \text{j}{{\bm{\tilde{C}}}^{(2)}}{{\bm{\tilde{C}}}^{(2)\text{H}}} \big) \big], \\
			&\!\!\!\!\!{{\bm{F}}_{\bm{\omega}\! \mathfrak{R}[\!{{\bm{\delta }}^{(2)}}\!]}}  \!=\! \frac{2}{\!\sigma_{\!\tilde{e}^{(2)}}}\! \mathfrak{R} \!\big[ \bm{\varLambda}_{1}^{(2)\text{H}} \bm{\tilde{C}}_{\bm{\omega}}^{(2)\text{H}} {{\bm{\tilde{C}}}^{(2)}}, \ldots, \bm{\varLambda}_{Q}^{(2)\text{H}} \bm{\tilde{C}}_{\bm{\omega}}^{(2)\text{H}}\! {{\bm{\tilde{C}}}^{(2)}} \big], \\ 
			&\!\!\!\!\!{{\bm{F}}_{\bm{\omega} \!\mathfrak{I}[\!{{\bm{\delta }}^{\!(2)}}\!]}} \!=\! \frac{2}{\!\sigma_{\!\tilde{e}^{(2)}}}\! \mathfrak{R}\! \big[ \text{j}\bm{\varLambda}_{1}^{(2)\text{H}} \bm{\tilde{C}}_{\bm{\omega}}^{(2)\text{H}} {{\bm{\tilde{C}}}^{(2)}}\!, \ldots, \text{j}\bm{\varLambda}_{Q}^{(2)\text{H}} \bm{\tilde{C}}_{\bm{\omega}}^{(2)\text{H}} \!{{\bm{\tilde{C}}}^{(2)}}\! \big] .
		\end{align}
		Applying the block matrix inversion property, we obtain
		\begin{align}
			\bm{CRB}_{\bm{\omega} \bm{\omega}}= \big[{{\bm{F}}_{\bm{\omega} \bm{\omega}}} - {{\bm{F}}_{\bm{\omega} {{\bm{\delta }}^{(2)}}}} {{\bm{F}}_{{{\bm{\delta }}^{(2)}}{{\bm{\delta }}^{(2)}}}} {{\bm{F}}^{\text{H}}_{\bm{\omega} {{\bm{\delta }}^{(2)}}}}\big]^{-1}.
		\end{align}
		\subsection{\textcolor{black}{Derivation of \textbf{{CRB}}$_{\tau\tau}$} }
		For \eqref{CChat}, by vectorizing $ {{\bm{\hat{C}}}^{(3)}}$, we obtain
		\begin{align}
			\text{vec}({{\bm{\hat{C}}}^{(3)}})=
			{{\bm{\hat{c}}}^{(3)}} = \big[ {{\bm{I}}_{Q \times Q}} \otimes {{\bm{C}}^{(3)}} \big]{{\bm{\delta }}^{(3)}} + {{\bm{e}}^{(3)}}, \label{22BB}
		\end{align}
		where ${{\bm{\delta }}^{(3)}} = \operatorname{vec}\big( {{\bm{\varDelta }}^{(3)}} \big)$, ${{\bm{e}}^{(3)}} $, and ${{\bm{\hat{c}}}^{(3)}}$ can be viewed as a distribution with mean vector ${{\bm{\mu }}^{(3)}} = [ {{\bm{I}}_{Q \times Q}} \otimes {{\bm{C}}^{(3)}} ]{{\bm{\delta }}^{(3)}}$ and covariance matrix ${{\bm{H}}^{(3)}} = \mathbb{E}\{ {{\bm{e}}^{(3)}}{{\bm{e}}^{(3)\text{H}}} \} = \sigma_{e^{(3)}}^{2}{{\bm{I}}_{N_{bs}Q \times N_{bs}Q}}$.
		
		We construct the FIM with respect to $\bm{\tau}$ and $\bm{\beta}$ as
		\begin{align}
			{{\bm{F}}_{(3)}}(\bm{\tau}, \bm{\beta}) =  \begin{bmatrix}
				{{\bm{F}}_{\bm{\tau} \bm{\tau}}} & {{\bm{F}}_{\bm{\tau} \bm{\beta}}} \\
				{{\bm{F}}^{\text{H}}_{\bm{\tau} \bm{\beta}}}& {{\bm{F}}_{\bm{\beta} \bm{\beta}}} \\
			\end{bmatrix} \in \mathbb{C}^{3Q \times 3Q}.
		\end{align}
		Since $\bm{\beta}$ is a complex-valued vector, we need to take derivatives with respect to its real and imaginary parts separately. Thus,
		\begin{align}
			{{\bm{F}}_{\bm{\beta} \bm{\beta}}} & =  \begin{bmatrix}
				{{\bm{F}}_{\mathfrak{R}[\bm{\beta}]\mathfrak{R}[\bm{\beta}]}} & {{\bm{F}}_{\mathfrak{R}[\bm{\beta}]\mathfrak{I}[\bm{\beta}]}} \\
				{{\bm{F}}^{\text{H}}_{\mathfrak{R}[\bm{\beta}]\mathfrak{I}[\bm{\beta}]}} & {{\bm{F}}_{\mathfrak{I}[\bm{\beta}]\mathfrak{I}[\bm{\beta}]}} \\
			\end{bmatrix}\in \mathbb{C}^{2Q \times 2Q}, \\
			{{\bm{F}}_{\bm{\tau} \bm{\beta}}} & = \big[ {{\bm{F}}_{\bm{\tau} \mathfrak{R}[\bm{\beta}]}}, {{\bm{F}}_{\bm{\tau} \mathfrak{I}[\bm{\beta}]}} \big] .
		\end{align}
		The first-order derivative of ${{\bm{\mu }}^{(3)}}$ with respect to $\bm{\tau}$ is given by
		\begin{align}
			\frac{\partial {{\bm{\mu }}^{(3)}}}{\partial \bm{\tau}} = {{\big[ \bm{\varLambda}_{1}^{(3)\text{H}} \bm{C}_{\bm{\tau}}^{(3)\text{H}}, \ldots, \bm{\varLambda}_{Q}^{(3)\text{H}} \bm{C}_{\bm{\tau}}^{(3)\text{H}} \big]}^{\text{T}}},
		\end{align}
		where $\bm{\varLambda}_{q}^{(3)} = {
			\diag}( {{[{{\bm{\varDelta }}^{(3)}}]}_{(:,q)}} )$. The matrix $\bm{C}_{\bm{\tau}}^{(3)}$ consists of the derivatives of the columns of $\bm{C}^{(3)}$ with respect to the corresponding $\bm{\tau}$, and is given by
		\begin{align}
			\bm{C}_{\bm{\tau}}^{(3)} = \Big[ {{\big. \frac{\partial {{\bm{c}}^{(3)}}}{\partial \tau} \big|}_{\tau = \tau_{1}}}, \ldots, {{\big. \frac{\partial {{\bm{c}}^{(3)}}}{\partial \tau} \big|}_{\tau = \tau_{Q}}} \Big] \in \mathbb{C}^{N_{\text{bs}} \times Q}.
		\end{align}
		Using \eqref{22BB}, we obtain the following expression:
		\begin{align}
			\!\!{{\bm{F}}_{\bm{\tau} \bm{\tau}}}= \frac{2}{\sigma_{e^{(3)}}^{2}} \mathfrak{R} \big[ \sum_{q=1}^{Q} \bm{\varLambda}_{q}^{(3)\text{H}} \bm{C}_{\bm{\tau}}^{(3)\text{H}} \bm{C}_{\bm{\tau}}^{(3)} \bm{\varLambda}_{q}^{(3)} \big].
		\end{align}
		Considering derivatives with respect to the real and imaginary parts of $\bm{\beta}$, we obtain:
		\begin{align}
			\frac{\partial {{\bm{\mu }}^{(3)}}}{\partial \mathfrak{R}[\bm{\beta}]} & = {{\big[ \bm{\varLambda}_{1}^{(3)\text{H}} \bm{C}_{\mathfrak{R}[\bm{\beta}]}^{(3)\text{H}}, \ldots, \bm{\varLambda}_{Q}^{(3)\text{H}} \bm{C}_{\mathfrak{R}[\bm{\beta}]}^{(3)\text{H}} \big]}^{\text{T}}}, \\
			\frac{\partial {{\bm{\mu }}^{(3)}}}{\partial \mathfrak{I}[\bm{\beta}]} & = {{\big[ \bm{\varLambda}_{1}^{(3)\text{H}} \bm{C}_{\mathfrak{I}[\bm{\beta}]}^{(3)\text{H}}, \ldots, \bm{\varLambda}_{Q}^{(3)\text{H}} \bm{C}_{\mathfrak{I}[\bm{\beta}]}^{(3)\text{H}} \big]}^{\text{T}}},
		\end{align}
		where $\bm{C}_{\mathfrak{R}[\bm{\beta}]}^{(3)}$ and $\bm{C}_{\mathfrak{I}[\bm{\beta}]}^{(3)}$ consist of the derivatives of the columns of $\bm{C}^{(3)}$ with respect to the corresponding $\mathfrak{R}[\bm{\beta}]$ and $\mathfrak{I}[\bm{\beta}]$, respectively, and are given by
		\begin{align}
			\bm{C}_{\mathfrak{R}[\bm{\beta}]}^{(3)} & = \Big[ {{\big. \frac{\partial {{\bm{c}}^{(3)}}}{\partial \mathfrak{R}[\beta]} \big|}_{\beta = \beta_{1}}}, \ldots, {{\big. \frac{\partial {{\bm{c}}^{(3)}}}{\partial \mathfrak{R}[\beta]} \big|}_{\beta = \beta_{Q}}} \Big] \in \mathbb{C}^{N_{\text{bs}} \times Q}, \\
			\bm{C}_{\mathfrak{I}[\bm{\beta}]}^{(3)} & = \Big[ {{\big. \frac{\partial {{\bm{c}}^{(3)}}}{\partial \mathfrak{I}[\beta]} \big|}_{\beta = \beta_{1}}}, \ldots, {{\big. \frac{\partial {{\bm{c}}^{(3)}}}{\partial \mathfrak{I}[\beta]} \big|}_{\beta = \beta_{Q}}} \Big] \in \mathbb{C}^{N_{\text{bs}} \times Q}.
		\end{align}
		Thus, we obtain
		\begin{align}
			{{\bm{F}}_{\mathfrak{R}[\bm{\beta}]\mathfrak{R}[\bm{\beta}]}} & = \frac{2}{\sigma_{e^{(3)}}^{2}} \mathfrak{R} \big[ \sum_{q=1}^{Q} \bm{\varLambda}_{q}^{(3)\text{H}} \bm{C}_{\mathfrak{R}[\bm{\beta}]}^{(3)\text{H}} \bm{C}_{\mathfrak{R}[\bm{\beta}]}^{(3)} \bm{\varLambda}_{q}^{(3)} \big], \\
			{{\bm{F}}_{\mathfrak{I}[\bm{\beta}]\mathfrak{I}[\bm{\beta}]}} & =  \frac{2}{\sigma_{e^{(3)}}^{2}} \mathfrak{R} \big[ \sum_{q=1}^{Q} \bm{\varLambda}_{q}^{(3)\text{H}} \bm{C}_{\mathfrak{I}[\bm{\beta}]}^{(3)\text{H}} \bm{C}_{\mathfrak{I}[\bm{\beta}]}^{(3)} \bm{\varLambda}_{q}^{(3)} \big], \\
			{{\bm{F}}_{\mathfrak{R}[\bm{\beta}]\mathfrak{I}[\bm{\beta}]}} & = \frac{2}{\sigma_{e^{(3)}}^{2}} \mathfrak{R} \big[ \sum_{q=1}^{Q} \bm{\varLambda}_{q}^{(3)\text{H}} \bm{C}_{\mathfrak{R}[\bm{\beta}]}^{(3)\text{H}} \bm{C}_{\mathfrak{I}[\bm{\beta}]}^{(3)} \bm{\varLambda}_{q}^{(3)} \big], 
		\end{align}
		\begin{align}
		 {{\bm{F}}_{\tau \mathfrak{R}[\bm{\beta}]}}& = \frac{2}{\sigma_{e^{(3)}}^{2}} \mathfrak{R} \big[ \sum_{q=1}^{Q} \bm{\varLambda}_{q}^{(3)\text{H}} \bm{C}_{\tau}^{(3)\text{H}} \bm{C}_{\mathfrak{R}[\bm{\beta}]}^{(3)} \bm{\varLambda}_{q}^{(3)} \big], \\
			{{\bm{F}}_{\bm{\tau} \mathfrak{I}[\bm{\beta}]}} & = \frac{2}{\sigma_{e^{(3)}}^{2}} \mathfrak{R} \big[ \sum_{q=1}^{Q} \bm{\varLambda}_{q}^{(3)\text{H}} \bm{C}_{\bm{\tau}}^{(3)\text{H}} \bm{C}_{\mathfrak{I}[\bm{\beta}]}^{(3)} \bm{\varLambda}_{q}^{(3)} \big].
		\end{align}
		Applying the block matrix inversion property, we obtain
		\begin{align}
			\bm{CRB}_{\bm{\tau} \bm{\tau}} = \big[{{\bm{F}}_{\bm{\tau} \bm{\tau}}} - {{\bm{F}}_{\bm{\tau} \bm{\beta}}} \bm{F}_{\bm{\beta} \bm{\beta}}^{-1} {{\bm{F}}^{\text{H}}_{\bm{\tau} \bm{\beta}}}\big]^{-1}.
		\end{align}
		
		To this end, we have derived the expressions for $\bm{CRB}_{\bm{\vartheta}\bm{\vartheta}}$, $\bm{CRB}_{\bm{\varphi}\bm{\varphi}}$, $\bm{CRB}_{\bm{\tau}\bm{\tau}}$, and $\bm{CRB}_{\bm{\omega}\bm{\omega}}$, along with explanations for all submatrix parameters.
		\vspace{-0.5em}
		\section{Simplification of CRB Specific Expressions for each target}\label{Appendix B}
		In this {appendix}, we derive the CRB for each individual target. 
		For the case of a single target source, equations \eqref{11dot} and \eqref{22dot} can be simplified as
		\begin{align}
			\tilde{\bm{a}}_{\text{\text{re}},\vartheta}(\mathcal{X},\mathcal{Y},\vartheta,\varphi) &= -\text{j}\frac{2\uppi}{\lambda_c}\bm{D}_x\bm{a}_{\text{\text{re}}}(\mathcal{X},\mathcal{Y},\vartheta,\varphi), \label{start}\\
			\tilde{\bm{a}}_{\text{\text{re}},\varphi}(\mathcal{X},\mathcal{Y},\vartheta,\varphi) &= -\text{j}\frac{2\uppi}{\lambda_c}\bm{D}_y\bm{a}_{\text{\text{re}}}(\mathcal{X},\mathcal{Y},\vartheta,\varphi),
		\end{align}
		where $\bm{D}_x =\diag(x_1,\ldots,x_{N_{\text{re}}})$ and $\bm{D}_y = \diag(y_1,\ldots,y_{N_{\text{re}}})$. 
		\eqref{F11}, \eqref{F22} and \eqref{F12} can be simplified to 
		\begin{align}\label{fuu}
			&\!\!\!\!F_{\vartheta\!\vartheta}\! = \!\frac{2}{{\sigma _{{e^{(1)}}}^2}}\mathfrak{R} \big[ {{\delta ^{(1)*}}{{ {\frac{{\partial {{\bm{a}}_{\text{\text{re}}}^{\text{H}}}}}{{\partial \vartheta }}} }}\frac{{\partial {{\bm{a}}_{\text {\text{re}}}}}}{{\partial \vartheta }}{\delta ^{(1)}}} \big] \!= \!\frac{{8{\pi ^2}{{\big| {{\delta ^{(1)}}} \big|}^2}}}{{\lambda _c^2\sigma _{{e^{(1)}}}^2}}\sum\limits_{n = 1}^{{N_{\text{\text{re}}}}} {x_n^2},
			\\
			&\!\!\!\!{F_{\varphi \varphi }} \!=\! \frac{2}{{\sigma _{{e^{(1)}}}^2}}\!\mathfrak{R} \!\big[ {{\delta ^{(1)*}}{{ {\frac{{\partial  {{\bm{a}}^\text{H}_{\text{\text{re}}}}}}{{\partial \varphi }}} }}\frac{{\partial {{\bm{a}}_{\text{\text{re}}}}}}{{\partial \varphi }}{\delta ^{(1)}}} \big]\!=\!\frac{{8{\pi ^2}{{\big| {{\delta ^{(1)}}} \big|}^2}}}{{\lambda _c^2\sigma _{{e^{(1)}}}^2}}\sum\limits_{n = 1}^{{N_{\text{re}}}} {y_n^2}, \label{fuv}
			\\
			&\!\!\!\!{F_{\vartheta \!\varphi }} \!=\! \frac{2}{{\sigma _{{e^{(1)}}}^2}}\!\mathfrak{R} \!\big[ {{\delta ^{(1)*}}{{ {\frac{{\partial {{\bm{a}}^\text{H}_{\text{re}}}}}{{\partial \vartheta }}} }}\frac{{\partial {{\bm{a}}_{\text{\text{re}}}}}}{{\partial \varphi }}{\delta ^{(1)}}} \big]\!
			\!=\! \frac{8\uppi^2|{\delta}^{(1)}|^2}{\lambda_c^2\sigma_{\text{e}^{(1)}}^2}\!\sum_{n=1}^{N_{\text{\text{re}}}}\!x_n y_n\label{fvv}.
		\end{align}
		Substituting \eqref{fuu}-\eqref{fvv} into \eqref{CRBUU} and \eqref{CRBVV} yields 
		\begin{align}
			\text{CRB}_{\vartheta\vartheta} &= \Big[\frac{8\uppi^2|\delta^{(1)}|^2}{\lambda_c^2\sigma_{\text{e}^{(1)}}^2}\big(\var(\mathcal{Y}) - \frac{\cov(\mathcal{X},\mathcal{Y})}{\var(\mathcal{X})}\big)\Big]^{-1} ,\\
			\text{CRB}_{\varphi\varphi} &= \Big[\frac{8\uppi^2|\delta^{(1)}|^2}{\lambda_c^2\sigma_{\text{e}^{(1)}}^2}\big(\var(\mathcal{X}) - \frac{\cov(\mathcal{X},\mathcal{Y})}{\var(\mathcal{Y})}\big)\Big]^{-1}.\label{end}
		\end{align}
		
		Similarly, by applying the method of  \eqref{start}-\eqref{end}, we can obtain \eqref{CRBnunu} and \eqref{CRBtautau}, respectively.

		\vspace{-0.5em}
		\section{Proof of Theorem \ref{Thoerem1}}
		\label{Appendix C}
		In this appendix, we consider sequentially placing $\mathcal{T}(m)$ for $m=1,\ldots,M$ along a finite time axis. Each placement is optimized by establishing an iterative relationship between the current objective function and the next placement.
		
		\begin{figure}[t]
			\vspace{-0.5em}
		\centering
		\begin{minipage}{0.65\linewidth}
			\centering
			\includegraphics[width=\linewidth]{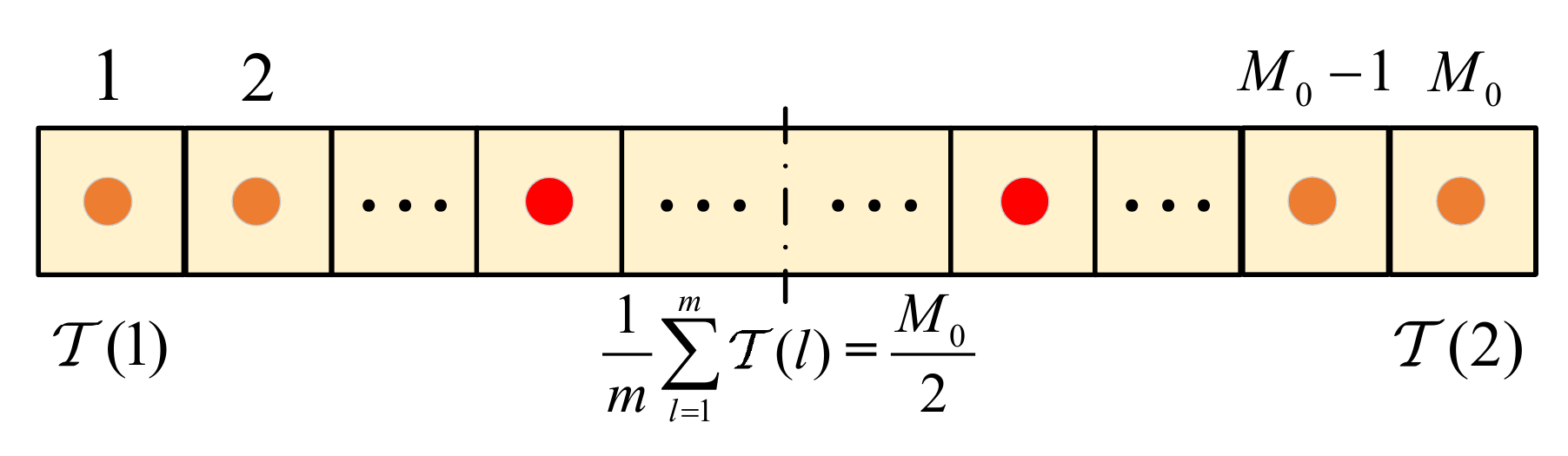}
			\subcaption{$m$ is even.}
			\label{Fig_oushu}
			\vspace{-0.5em}
		\end{minipage}
		\begin{minipage}{0.65\linewidth}
			\centering
			\includegraphics[width=\linewidth]{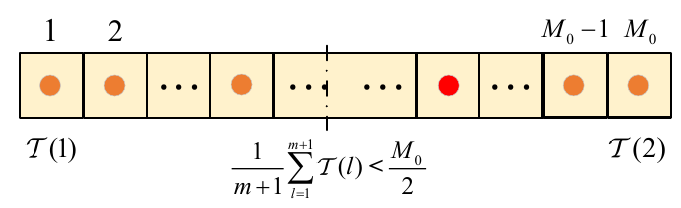}
			\subcaption{$m$ is odd.}
			\label{Fig_jishu}
		\end{minipage}
		\vspace{-0.5em}
		\caption{Illustration of the optimized design of OFDM symbols.}
		\label{Fig_optimization}
		\vspace{-0.75cm}
		\end{figure}
		
	We conceptualize a strip of length $M_0$ where each box represents a unit. The optimization problem is to place $\mathcal{T}$ to maximize the objective function $\var(\mathcal{T})$.
		
	Clearly, to maximize the objective function, the first two placements must be at the extremes $\mathcal{T}(1) = 0$ and $\mathcal{T}(2) = M_0$. Assuming $m$ points have been placed, the variance is given by
		\begin{align}
			\mathrm{var}^m = \frac{1}{m}\sum_{l=1}^m \mathcal{T}(l)^2 - \big(\frac{1}{m}\sum_{l=1}^m \mathcal{T}(l)\big)^2.
		\end{align}
		After placing the $(m+1)$-th point, the objective function becomes
		\begin{align}
			\mathrm{var}^{m+1} = \frac{1}{m+1}\sum_{l = 1}^{m+1} \mathcal{T}(l)^2 - \big(\frac{1}{m+1}\sum_{l = 1}^{m+1} \mathcal{T}(l)\big)^2.
		\end{align}
		It should be noted that
		\begin{align}
			\!\mathrm{var}^{m\!+\!1} \!=\! \frac{m}{m\!+\!1}\mathrm{var}^m \!+\! \frac{m}{(m\!+\!1)^2}\big( \mathcal{T}{(m\!+\!1)} \!-\! \frac{1}{m}\!\sum_{l = 1}^m \mathcal{T}(l) \big)^2.
		\end{align}
		For the $(m+1)$-th placement, the term $\frac{m}{m+1}\mathrm{var}^m$ is a fixed constant. To maximize the objective function, we must maximize the squared distance term, which requires placing the next OFDM symbol at the available position farthest from the current mean position $\frac{1}{m}\sum_{l = 1}^m \mathcal{T}(l)$. 
		
	Without loss of generality, we assume $m$ is even. As shown in Fig. \ref{Fig_oushu}, the mean position is exactly at the center $M_0/2$. To maximize the distance from the mean while satisfying the separation constraint $\mathcal{T}(m_1) \neq \mathcal{T}(m_2)$ for all distinct pairs, the $(m+1)$-th OFDM symbol must be placed at the available positions marked in red. Since both candidate positions are symmetric about $M_0/2$, we select the left position without loss of generality. After this placement, the number of currently placed points becomes odd, and the new mean position strictly shifts to the left. Consequently, as shown in Fig. \ref{Fig_jishu}, to maximize the variance for the $(m+2)$-th placement, the next OFDM symbol must unequivocally be placed at the rightmost available position.
		
	By recursively applying this alternating placement strategy, we guarantee that every single sequential placement strictly maximizes the variance increment. This recursive mathematical approach rigorously ensures that the final distribution of the $M$ OFDM symbols achieves the global optimality. Repeating this process yields the optimal solution
		\begin{align}
			\mathcal{T}(m)^{\star} = \begin{cases}
				1,2,\ldots,\bar{M} & m \leq \bar{M}, \\
				M_0 - \bar{M} + 1,\ldots,M_0 & m > \bar{M},
			\end{cases}
		\end{align}
		where $\bar{M} = \lfloor M/2 \rfloor$.
		
		To this end, we have completed the proof of Theorem \ref{Thoerem1}.

		
		
		%
			\vspace{-0.5em}
		\bibliographystyle{IEEEtran}
		\bibliography{mybib_Abbreviation_1}

	\end{document}